\documentclass[11pt]{article}

\newif\ifpreprint   \preprinttrue
\newif\ifbiblatex   \biblatexfalse

\usepackage{amsthm,amsmath,amssymb,mathtools}
\usepackage{graphicx,booktabs}
\usepackage[dvipsnames]{xcolor}
\usepackage{microtype}
\usepackage{subfig}
\usepackage{algorithm}
\usepackage{algpseudocode}
\usepackage{tikz}

\ifpreprint
  \usepackage[margin=1in]{geometry}
  \usepackage{authblk}
  \usepackage{newtxtext,newtxmath}   
\else
  \providecommand{\affil}[2][]{}     
\fi

\ifbiblatex
\else
  \usepackage{natbib}
  \let\parencite\citep
  \let\textcite\citet
\fi

\usepackage{my-macros}

\DeclareMathOperator{\Cov}{Cov}
\DeclareMathOperator{\var}{var}

\ifpreprint
  \usepackage[colorlinks=true,linkcolor=Maroon,
    urlcolor=MidnightBlue,citecolor=MidnightBlue]{hyperref}
\else
  \usepackage[colorlinks=false]{hyperref}
\fi
\usepackage[capitalise,noabbrev]{cleveref}

\providecommand{\keywords}[1]{\small\textbf{\textit{Keywords---}} #1}

\title{Experimental Design When N Equals One}

\author{Wenxuan Guo}
\author{Tengyuan Liang}

\affil{The University of Chicago}

\date{\today}

\begin{document}

\maketitle

\begin{abstract}
N-of-1 trials, or time-series experiments, are widely used in clinical research and online platforms. Yet the theoretically optimal design for estimating many treatment effects remains unclear. We propose a simple Markovian framework for experimental design in which the treatment assignment process is governed by possibly time-varying transition matrices. This formulation encompasses many existing N-of-1 designs and provides a principled way to control temporal dependence in treatment assignment through Markov transition probabilities. Under a finite-order impulse-response model, we formulate the design objective as minimizing the estimation error of ordinary least squares estimators for target treatment effects, and propose practical design optimization procedures. To characterize the optimal temporal structure, we focus on two structured design classes, random-switch and cycle-switch designs, and establish a complete large-$T$ asymptotic theory for the optimal designs in both classes. Our results justify the robustness of i.i.d. Bernoulli designs in N-of-1 trials and quantify how the optimal design depends on the target estimand, including cumulative and lag-specific treatment effects. Simulations demonstrate the effectiveness and robustness of the proposed designs across multiple scenarios.
\end{abstract}

\keywords{causal inference, Markov chain, experimental design, time series.}

\tableofcontents


\section{Introduction}\label{sec:intro}
N-of-1 trials, also referred to as time-series experiments, have been widely applied in clinical research \parencite{lillie2011n}, online platforms \parencite{bojinov2023design}, and other domains \parencite{mirza2017history,hawksworth2024review}. In such trials, treatment assignment is repeatedly randomized over time for a single experimental unit. N-of-1 trials allow treatment assignments and outcomes to evolve over time, making them particularly attractive in modern applications where repeated interventions are feasible, experimental units are limited, and the scientific objective is to estimate cumulative or lagged treatment effects. Moreover, N-of-1 trials are naturally suited for individualized treatment evaluation, allowing investigators to assess treatment effects at the level of a single experimental unit \parencite{bojinov2020avoid,liang2023randomization}.

Despite their growing popularity, the choice of optimal design in N-of-1 experiments remains to be elucidated. A central difficulty is that different target effects favor different forms of temporal dependence in the treatment assignment path. For example, independent Bernoulli assignment is commonly used for estimating lag-specific treatment effects \parencite{lin2025unifying,liang2023randomization}, whereas switchback designs, which hold treatment fixed over consecutive time windows, are often used for estimating cumulative or carryover effects. These examples suggest that design should be tailored to the target estimand, but they do not provide a general principle on experimental design. In particular, with a few exceptions \parencite{bojinov2023design,chen2025efficient}, optimal designs for a given estimand remain poorly characterized. The problem becomes even less clear when the target is a class of estimands, such as multiple lag-specific effects or cumulative effects over different time windows. Thus, there is a need for a general design framework that can flexibly target different treatment effects while providing optimality guarantees.

In this paper, we adopt a Markovian viewpoint and model the design of N-of-1 trials through Markov transition matrices. Specifically, we model the treatment assignment path $Z_1,\dots,Z_T$ as a Markov chain, with the transition from $Z_t$ to $Z_{t+1}$ governed by
\[
P_t = \begin{bmatrix}
    1 - \rho_t & \rho_t \\
    \gamma_t & 1 - \gamma_t \\
\end{bmatrix}\;, \quad \rho_t, \gamma_t \in [0, 1]\;.
\]
The two states of the Markov chain correspond to control and treatment, respectively. The Markovian viewpoint provides a characterization of the treatment assignment distribution, and encompasses the existing designs as special cases. For instance, by setting $\rho_t + \gamma_t = 1$, we recover the independent time-series experiments in \textcite{liang2023randomization} and \textcite{lin2025unifying} (Example \ref{ex:indep}). By enforcing a periodic structure on $P_t$, we recover the switchback experiments in \textcite{bojinov2023design,chen2025efficient} (Example \ref{ex:regular}).

To define the design objective and search for the optimal design, we follow the literature on dynamic treatment effects and temporal interference \parencite{granger1969,sims1972money,liang2023randomization} and adopt a finite-order impulse-response model. Specifically, we assume that outcomes depend additively on the current and past treatment assignments through
\begin{equation}\label{eq:model}
Y_t = \alpha + \sum_{k=1}^K \beta_k Z_{t-k+1} + \epsilon_t\;.
\end{equation}
Here, $\{Y_t\}_{t=K}^T$ are the (potentially detrended) outcomes, $\{\epsilon_t\}_{t=K}^T$ are i.i.d. errors, and $\beta_k$ captures the effect of a treatment assignment after $k-1$ periods. The impulse-response model enables the estimation of a large class of treatment effects from linear combinations. For a given estimand vector $w=(w_1,\dots,w_K)^\top \in \mathbb{R}^K$, we define the target treatment effect and its estimator as
\[
\tau_w = \sum_{k=1}^K w_k \beta_k\;,
\qquad \widehat{\tau}_w = \sum_{k=1}^K w_k \widehat{\beta}_k\;.
\]
Here, $\widehat{\beta}_k$, $k=1, \dots, K$ are ordinary least squares (OLS) estimates from Model \eqref{eq:model}. Then, our goal is to find the optimal design under the Markovian modeling that minimizes the estimation error $\var(\widehat{\tau}_w)$.

The Markovian framework, combined with the impulse-response model, offers several advantages for addressing the experimental design problem in N-of-1 trials. First, it allows experimenters to evaluate the estimation error---namely, $\var(\widehat{\tau}_w)$---of candidate designs before running the experiment. Second, by imposing structured restrictions on the candidate design class, the design problem can be formulated as an optimization problem over Markov parameters, yielding a tractable approach to design selection. Last, when the experimenter is interested in multiple treatment effects (through multiple $w$), the design selection procedure can be naturally extended to control multiple estimation variances through a minimax optimization criterion.

We illustrate the three advantages through a visualization of the design selection procedure. In Figure~\ref{fig:demo}, each row represents a candidate design parameterized by $\{(\rho_t,\gamma_t)\}_{t=1}^{T-1}$, and each column represents a target effect. The color of each cell indicates the estimation error for the corresponding design-target pair, with the smallest error in each column marked by a red circle. The final column aggregates the target-specific errors by their worst-case value, following the minimax design perspective \parencite{basse2023minimax,cox2000theory}.

\begin{figure}[h]
\centering
\begin{tikzpicture}[x=1cm,y=1cm]
\def\m{5}
\def\cellw{0.8}
\def\cellh{0.8}
\def\gap{1.3}
\def\dotgap{1.5}

\pgfmathsetmacro{\xA}{0}
\pgfmathsetmacro{\xB}{\xA+\cellw+\gap}
\pgfmathsetmacro{\xC}{\xB+\cellw+\gap}
\pgfmathsetmacro{\xDots}{\xC+\cellw+\dotgap}
\pgfmathsetmacro{\xD}{\xDots+\dotgap}
\pgfmathsetmacro{\xE}{\xD+\cellw+\gap}
\pgfmathsetmacro{\xF}{\xE+\cellw+\gap}

\node[font=\small] at (\xA+0.5*\cellw,0.5) {effect 1};
\node[font=\small] at (\xB+0.5*\cellw,0.5) {effect 2};
\node[font=\small] at (\xC+0.5*\cellw,0.5) {effect 3};
\node[font=\small] at (\xDots,0.5) {$\cdots$};
\node[font=\small] at (\xD+0.5*\cellw,0.5) {effect $J\!-\!1$};
\node[font=\small] at (\xE+0.5*\cellw,0.5) {effect $J$};
\node[font=\small] at (\xF+0.5*\cellw,0.5) {worst-case};

\newcommand{\drawcell}[3]{%
    \pgfmathsetmacro{\y}{-#2*\cellh}%
    \filldraw[draw=black, fill=blue!#3!white]
        (#1,\y) rectangle (#1+\cellw,\y+\cellh);
}

\def\vmin{11}
\def\vmax{72.8}
\foreach \i/\sA/\sB/\sC/\sD/\sE/\sF in {
    1/61.6/61.6/61.6/61.6/61.6/61.6,
    2/49.1/29.6/68.7/55.4/68.7/68.7,
    3/30.3/18.9/75.8/50.7/75.8/75.8,
    4/49.1/13.6/82.9/45.2/82.9/82.9,
    5/61.6/10.0/90.0/40.3/90.0/90.0
}{
    \pgfmathsetmacro{\y}{-\i*\cellh}

    \node[anchor=east,font=\small] at (-0.2,\y+0.5*\cellh)
        {design $\i$};

    \node[font=\large] at (\xDots,\y+0.5*\cellh) {$\cdots$};

    \drawcell{\xA}{\i}{\sA}
    \drawcell{\xB}{\i}{\sB}
    \drawcell{\xC}{\i}{\sC}
    \drawcell{\xD}{\i}{\sD}
    \drawcell{\xE}{\i}{\sE}
    \drawcell{\xF}{\i}{\sF}
}

\draw[thick] (\xA,-\m*\cellh) rectangle (\xA+\cellw,0);
\draw[thick] (\xB,-\m*\cellh) rectangle (\xB+\cellw,0);
\draw[thick] (\xC,-\m*\cellh) rectangle (\xC+\cellw,0);
\draw[thick] (\xD,-\m*\cellh) rectangle (\xD+\cellw,0);
\draw[thick] (\xE,-\m*\cellh) rectangle (\xE+\cellw,0);
\draw[thick] (\xF,-\m*\cellh) rectangle (\xF+\cellw,0);

\foreach \x/\r in {
    \xA/3,
    \xB/5,
    \xC/1,
    \xD/5,
    \xE/1,
    \xF/1
}{
    \pgfmathsetmacro{\yc}{-\r*\cellh + 0.5*\cellh}
    \draw[red,very thick] (\x+0.5*\cellw,\yc) circle[radius=0.6];
}

\pgfmathsetmacro{\yopt}{-3*\cellh + 0.5*\cellh}
\pgfmathsetmacro{\xright}{\xF+\cellw}
\end{tikzpicture}
\caption{Illustration of minimax design optimization. Each row represents a candidate design and each column represents an estimand. The omitted middle estimands are indicated by dots.}
\label{fig:demo}
\end{figure}
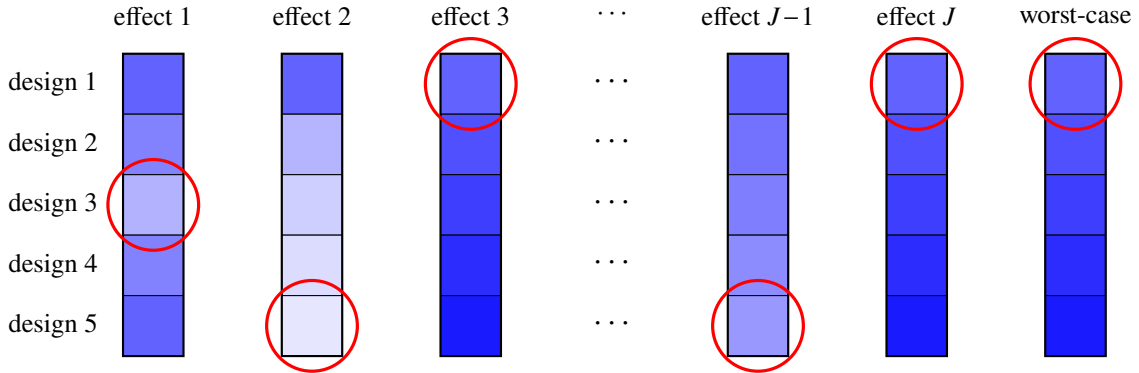

To gain insight into the temporal dependence structure of the optimal design, we further study the theoretical behavior of the optimal design parameters. Specifically, we focus on two structured classes of designs and examine optimal temporal dependence through two complementary features: the switching probabilities over time and the lengths of persistent treatment and control periods.

\begin{enumerate}
  \item \textbf{Optimal design theory for switching probabilities (Section \ref{sec:homog}):} We consider a structured class of \emph{random-switch designs} with a constant Markov transition matrix, that is, $\rho_t=\rho$ and $\gamma_t=\gamma$. Under such designs, the temporal treatment dependence is fully characterized by the switching probabilities $\rho$ and $\gamma$. We analyze the behavior of optimal design parameters $\rho^\star$ and $\gamma^\star$ for a variety of target treatment effects. 
  \item \textbf{Optimal design theory for persistent treatment cycles (Section \ref{sec:exact}):} We consider another structured class of \emph{cycle-switch designs} as deterministic treatment assignment processes in which the treatment status switches every $l$ time points. We characterize the optimal block length $l^\star$ for different target effects.
\end{enumerate}

Based on the theoretical analysis, Table \ref{tab:opt} highlights several preferred design patterns in terms of switching probabilities and block lengths. From the switching-probability perspective, the optimal switching probability is $1/2$ for estimating lag-specific treatment effects and for constructing robust designs, which corresponds to the i.i.d. Bernoulli(1/2) randomization. By contrast, for cumulative treatment effects, the optimal switching probability moves to the boundary, leading to ill-behaved limiting designs. From the cycle-length perspective, the optimal block length is approximately $K+\sqrt{K}$ for cumulative treatment effects, while it is $K$ for estimating lag-specific effects and for robust designs. The optimal design parameters for general target estimands are also characterized in the paper.

\begin{table}[h]
  \centering
  \begin{tabular}{c|c|c|c}
    \hline
     & cumulative effect & lag-specific effect & robust \\
    \hline
    switching probability & - & $\rho^\star = \gamma^\star = 1/2$ &  $\rho^\star = \gamma^\star = 1/2$ \\
    block length & $l^\star \approx K-1+\sqrt{K-1}$ & $l^\star \approx K$ & $l^\star = K$\\
    \hline
  \end{tabular}
  \caption{Informal summary of optimal design parameters for different estimation targets. $(\rho, \gamma)$ and $l$ are the design parameters under random-switch and cycle-switch designs. $K$ is the total number of lags in Model \eqref{eq:model}. The switching-probability results follow from Example \ref{ex:lag} and Theorem \ref{thm:robust_random}, and the block-length results follow from Corollaries \ref{cor:cumulative} and \ref{cor:lag}, and Theorem \ref{thm:robust2}.}
  \label{tab:opt}
\end{table}

In Section \ref{sec:simu}, we implement the proposed designs and compare their numerical performance with state-of-the-art designs for N-of-1 trials and switchback experiments. The results demonstrate the effectiveness of our procedures under the impulse-response model. In addition, under model misspecification, random-switch designs exhibit robustness and provide reasonable treatment effect estimates across multiple practical scenarios.

\subsection{Related Work}

In longitudinal studies, Robins and his collaborators developed the $g$-methods framework for causal inference with time-varying treatments and covariates \parencite{robins1986new,robins1999estimation,robins2000marginal}, allowing current outcomes to depend on past treatments, outcomes, and covariates; see also \textcite{zeger1986longitudinal,blackwell2018make,hernan2020causal}. A related line of work studies dynamic treatment regimes and adaptive treatment strategies, where treatment decisions are allowed to depend on an individual's evolving history \parencite{murphy2003optimal,robins2004optimal,murphy2001marginal,murphy2005experimental}. The causal effect in such studies is usually a population average over a cohort of experimental units. Since the number of possible treatment paths and counterfactual outcomes grows exponentially over time, parsimonious models are often introduced to make estimation and inference feasible; our impulse-response model \eqref{eq:model} follows the same principle.

Recent years have witnessed a growing literature on the design and analysis of N-of-1 trials and time-series experiments. Unlike classical longitudinal studies, these experiments typically observe a single experimental unit over a potentially long time horizon, and much of the existing analysis is design-based. \textcite{bojinov2019time} studied a general class of causal effects and developed design-based methods for estimation and inference. Building on this work, \textcite{bojinov2023design,chen2025efficient} focused on cumulative treatment effects and studied optimal regular switchback experiments under different causal mechanisms. \textcite{lin2025unifying} estimated lag-specific treatment effects using regression-based estimators, while \textcite{liang2023randomization} studied a broad class of causal estimands under a linear impulse-response model and developed method-of-moments estimators that allow arbitrary noise structures. Both works suggest the robustness of independent treatment assignments in N-of-1 trials. Our work is model-based and complements the design-based approaches discussed above. This distinction reflects a tradeoff between generality and interpretability. By focusing on the impulse-response model, the treatment effects reduce to linear functionals of model parameters that are easy to interpret. Moreover, the model enables a unified design optimization procedure that accommodates a broad range of estimands, including cumulative and lag-specific treatment effects. In addition, Remark~\ref{rmk:est} and Section~\ref{sec:connection} provide a design-based interpretation of our approach, connecting it to existing work such as \textcite{lin2025unifying}.

The study of optimal experimental design has a long history. Early work focused on specified outcome models and introduced optimality criteria such as D-, G-, and E-optimality \parencite{kiefer1959optimum,cox2000theory}. In particular, E-optimality minimizes the worst-case variance over a class of linear combinations of model coefficients, so our minimax design objective can be viewed as a generalization of E-optimality to N-of-1 trials. The outcome models considered in this literature are often formulated as response surfaces, with low-order polynomials of experimental factors used as regressors \parencite{box1959basis,dean2015handbook}. More recent work has adapted classical optimal design ideas to modern online experimentation; for example, \textcite{bhat2020near} used $D_A$-optimality, a variant of D-optimality, to construct near-optimal designs for online A/B experiments.

A complementary line of work studies optimal experimental design from a model-agnostic and design-based perspective. \textcite{harshaw2024balancing} characterized the tradeoff between covariate balance and robustness and proposed Gram-Schmidt Walk designs to navigate this tradeoff. \textcite{guo2025gauss} studied optimal covariate balance by modeling treatment assignments through latent Gaussian variables, thereby converting the design problem into an optimization problem in Gaussian space. See \textcite{zhao2024experimental} for a comprehensive review of related work. Compared with this literature, our work is model-based in the spirit of response surface methodology, but focuses on lagged treatment assignments as regressors for a single experimental unit.

\section{Parametrization of the Design Space}\label{sec:param}
Throughout the paper, we use the uppercase letter \(Z_t\) to denote the treatment assignment generated by an experimental design. We use the lowercase letter \(z_t\in\{0,1\}\) to denote a fixed treatment value.

In this section, we parametrize the space of probability distributions in designing time-series experiments using Markov transition probabilities. This formulation yields a flexible class of designs that includes independent assignments, regular switchback experiments, and deterministic switching as special cases. We define this class as follows. 

\begin{definition}\label{def:markov}
The treatment path $\{Z_t\}_{t=1}^T$ is said to follow a Markovian assignment design with parameters $\{(\rho_t,\gamma_t)\}_{t=1}^{T-1}$ if it is a Markov chain on $\{0,1\}$ with transition matrix
\[
P_t =
\begin{bmatrix}
    1-\rho_t & \rho_t \\
    \gamma_t & 1-\gamma_t
\end{bmatrix},
\qquad t=1,\ldots,T-1\;.
\]
\end{definition}
In Definition~\ref{def:markov}, we leave the initial treatment probability, $\Pr(Z_1=1)$, unrestricted. We specify the initial treatment probability later when necessary.
The parameters $\rho_t$ and $\gamma_t$ determine the switching probabilities, since
\[
\Pr(Z_{t+1}=1\mid Z_t=0)=\rho_t\;,\qquad
\Pr(Z_{t+1}=0\mid Z_t=1)=\gamma_t \;.
\]
Different choices of $(\rho_t,\gamma_t)$ recover several treatment-assignment strategies commonly used in N-of-1 trials. We give a few examples below.
\begin{example}[Independent treatment assignment]\label{ex:indep}
If $\rho_t+\gamma_t=1$, then $\gamma_t=1-\rho_t$ and
\[
P_t =
\begin{bmatrix}
    1-\rho_t & \rho_t \\
    1-\rho_t & \rho_t
\end{bmatrix}\;.
\]
In this case, the transition probabilities do not depend on the current treatment status. Hence the chain forgets its history after one step, and $Z_{t+1}$ is assigned independently with treatment probability $\rho_t$. This recovers the independent treatment assignments studied by \textcite{liang2023randomization,lin2025unifying}.
\end{example}

\begin{example}[Regular switchback experiments]\label{ex:regular}
Let $1=t_0<t_1<\cdots<t_M\le T$ denote the randomization times, and let $(q_0,\ldots,q_M)\in(0,1)^{M+1}$. Set $\Pr(Z_1=1)=q_0$. For
$t=1,\dots,T-1$, define
\[
P_t =
\begin{bmatrix}
    1 & 0 \\
    0 & 1
\end{bmatrix}
\qquad
\text{if } t\notin\{t_1-1,\ldots,t_M-1\}\;,\qquad P_{t_m-1} =
\begin{bmatrix}
    1-q_m & q_m \\
    1-q_m & q_m
\end{bmatrix},
\qquad m=1,\ldots,M\;.
\]
Under this specification, the treatment status is held fixed between successive randomization times, while at each randomization time $t_m$ it is redrawn independently with treatment probability $q_m$. Thus the Markov assignment design
recovers regular switchback experiments \parencite{bojinov2023design,chen2025efficient}. These designs can be viewed as cluster randomization \parencite{fisher1926arrangement} in the time dimension, with each time window forming a
cluster.
\end{example}

\begin{example}[Deterministic designs]\label{ex:deterministic}
Two boundary choices of the Markov parameters yield deterministic assignment rules. If $\rho_t=\gamma_t=1$, then $P_t =
\begin{bmatrix}
    0 & 1 \\
    1 & 0
\end{bmatrix}$, so the treatment status switches deterministically at every time point. This is the alternating design. Conversely, if $\rho_t=\gamma_t=0$, then $P_t =
\begin{bmatrix}
    1 & 0 \\
    0 & 1
\end{bmatrix}$,
so the treatment status is held fixed and never switches. Thus deterministic alternation and deterministic persistence arise as two limiting cases of the Markovian assignment design.
\end{example}

The Markovian viewpoint yields closed-form expressions for the moments of the treatment-assignment process. Define the product kernel
\[
A(s,t) \coloneqq
\begin{cases}
\prod_{j=s}^t (1-\gamma_j-\rho_j), & \text{if } s \leq t, \\
1, & \text{if } s > t .
\end{cases}
\]
The following proposition characterizes the first and second moments of the design. These moment expressions will be used in Sections~\ref{sec:homog} and~\ref{sec:exact} to formulate design objectives. The proof is given in Section~\ref{sec:proof} of the Appendix.

\begin{proposition}\label{prop:moments}
For any $1 \leq s < t \leq T$,
\begin{align*}
    \Pr(Z_s=1) &= A(1, s-1) \Pr(Z_1 = 1)+ \sum_{i=1}^{s-1} A(i+1, s-1) \rho_i\;, \\
    \Pr(Z_t=1 \mid Z_s=1) &= A(s, t-1) + \sum_{i=s}^{t-1} A(i+1, t-1) \rho_i \;.
\end{align*}
\end{proposition}
Proposition~\ref{prop:moments} immediately yields the required first and second moments since
\[
  \E Z_s = \Pr(Z_s=1)\;, \qquad
  \E Z_s Z_t = \Pr(Z_s=1)\Pr(Z_t=1\mid Z_s=1)\;.
\]
In the next section, we introduce the finite-order impulse-response model. Under this model, the estimation variance can be written in terms of these treatment-assignment moments. Proposition~\ref{prop:moments} therefore provides the main building block for deriving tractable design optimization criteria in our work.

\section{Estimand, Estimator, and Design Objective}
In this section, we specify the outcome model and introduce design objectives that quantify estimation error under the working model. Following \textcite{liang2023randomization}, we consider the finite-order impulse-response model
\begin{equation}\label{eq:main_model}
Y_t  = \alpha + \sum_{k=1}^K \beta_k Z_{t-k+1} + \epsilon_t\;,
\qquad
\epsilon_t \stackrel{\mathrm{iid}}{\sim} \mathcal{N}(0,\sigma^2)\;, 
\qquad 
t = K, \dots, T\;.
\end{equation}
Here, $\alpha$ denotes the baseline mean, $\beta_k$ is the effect of treatment received at lag $k-1$, and $\epsilon_t$ is an idiosyncratic error term. We assume that the errors are independent of the treatment-assignment process, so that treatment assignments are exogenous and do not depend on past outcomes. Throughout the paper, we treat the lag length $K$ as fixed while allowing the time horizon $T$ to grow to infinity.

The model specification above departs from the design-based approach commonly used in N-of-1 trials \parencite{liang2023randomization,bojinov2019time,bojinov2023design,lin2025unifying}. This departure reflects a tradeoff between generality and interpretability. On the one hand, design-based approaches are nonparametric and can deliver valid estimators under weak outcome-model assumptions. On the other hand, the resulting estimands in N-of-1 trials often take the form of complex time averages of potential outcomes, and commonly used estimands can be design-dependent \parencite{bojinov2019time,lin2025unifying}. Our model-based approach sacrifices some generality in causal effect estimation, but yields simpler and design-invariant estimands. It also provides a unified framework for design optimization in N-of-1 trials, as described in Section \ref{sec:intro}. In addition, in Remark~\ref{rmk:est} we provide a design-based interpretation of our approach.

The model-based approach remains connected to the Neyman-Rubin potential outcome framework \parencite{neyman1923,rubin1974estimating}. Let $Y_t(z_1,\ldots,z_T)$ denote the potential outcome at time $t$ under the full treatment path
$(z_1,\ldots,z_T)\in\{0,1\}^T$. Our model imposes that
\[
Y_t(z_1, \dots, z_T) = Y_t(z_{t-K+1},\ldots,z_t) =
\alpha+\sum_{k=1}^K \beta_k z_{t-k+1}+\epsilon_t\;.
\]
This representation makes two restrictions explicit. First, the model imposes finite carryover: the outcome at time $t$ depends only on the current and previous $K-1$ treatment assignments. This assumption is common in the N-of-1 literature
\parencite{bojinov2019time,bojinov2023design,chen2025efficient}. Infinite carryover has also been studied \parencite{liang2023randomization,hu2025geometric}, and here we focus on finite carryover to obtain tractable design optimizations. Second, the model satisfies non-anticipation: the potential outcome at time $t$ does not depend on future treatment assignments \parencite{bojinov2019time,liang2023randomization,lin2025unifying}.

\begin{remark}[Model misspecification and detrended outcomes]
To reduce the risk of model misspecification, we interpret $y_t$ as a detrended outcome of interest, obtained after adjusting for predictable time variation. Such adjustment is feasible, for example, when pre-experiment historical outcomes are available, as is often the case in modern online experiments \parencite{xiong2024data}. In addition, our numerical experiments in Section~\ref{sec:simu} show that the optimal designs obtained under our framework exhibit robustness to several forms of model misspecification.
\end{remark}

\subsection{Estimand and Estimator}
Given Model \eqref{eq:main_model}, we consider treatment effects that are linear combinations of the impulse-response coefficients. For a given vector $w\in\mathbb{R}^K$, define
\[
\tau_w = \sum_{k=1}^K w_k\beta_k = w^\top\beta \;.
\]
This class includes meaningful causal effects for N-of-1 trials. Throughout the paper, we focus primarily on the following two choices of $w$.
\begin{example}[Cumulative treatment effects]
If $w=(1,\ldots,1)^\top$, then $\tau_w=\sum_{k=1}^K \beta_k$ is the cumulative effect of treatment over the current period and the previous $K-1$ periods. Under Model~\eqref{eq:main_model}, it is the contrast between sustained treatment and sustained control over a length-$K$ treatment history. Cumulative effects serve as a key causal quantity to study in longitudinal experiments \parencite{robins1986new} and modern time-series experiments \parencite{bojinov2023design,chen2025efficient}.
\end{example}

\begin{example}[Lag-specific treatment effects]
Let $e_k$ denote the $k$-th canonical basis vector in $\mathbb{R}^K$. Setting
$w=e_k$ gives $\tau_w=\beta_k$, which captures the effect of treatment received at lag $k-1$ on the current outcome. This corresponds to the lag-specific treatment effect studied by \textcite{lin2025unifying}.
\end{example}

We estimate $\tau_w$ using ordinary least squares (OLS). Based on Model \eqref{eq:main_model}, define the lagged treatment vector
\[
X_t=(Z_t,Z_{t-1},\dots,Z_{t-K+1})^\top \in \R^K\;,
\qquad t=K,\dots,T\;.
\]
Let $T_\text{eff}=T-K+1$ denote the effective number of observations after forming the $K$ lagged treatment variables, and define $\bar{X} = \sum_{t=K}^T X_t / T_\text{eff}$. Then, the OLS estimator for $\beta$ can be expressed as
\begin{equation}\label{eq:bhat}
\widehat\beta =
\Bigl( \frac1{T_\text{eff}}\sum_{t=K}^T (X_t-\bar X)(X_t-\bar X)^\top \Bigr)^{-1} \Bigl( \frac1{T_\text{eff}}\sum_{t=K}^T (X_t-\bar X)Y_t \Bigr)\;.
\end{equation}
The corresponding OLS estimator for $\tau_w$ is
\[
\widehat{\tau}_w = w^\top \widehat{\beta}\;.
\]

\begin{remark}[Interpretation of $\widehat{\beta}$ under design-based inference]\label{rmk:est}
From a design-based perspective, the OLS estimator need not rely on the linear outcome model as a literal data-generating model. Without imposing such a model, the OLS estimator can instead be interpreted as estimating a best linear projection of the observed outcomes onto the treatment history. More concretely, in Section \ref{sec:connection}, we show that under design-based inference, $\widehat{\beta}$ converges to a vector of conditional-outcome contrasts induced by the experimental design. Moreover, under i.i.d. Bernoulli randomization or linear additive potential outcomes, this limiting projection target coincides with the lagged treatment effect estimand studied by \textcite{lin2025unifying}. Thus, even beyond the model-based interpretation, the OLS estimator retains a design-based causal interpretation. We refer readers to Section \ref{sec:connection} of the Appendix for further details.
\end{remark}

\begin{remark}[Alternative design-based estimators]
Our Markovian designs can also be used with other design-based estimands for time-series experiments. For example, the Horvitz-Thompson-type estimators of \textcite{bojinov2019time} can be computed directly under our random-switch designs by using the corresponding Markov transition probabilities. When the transition probabilities are bounded away from zero and one, these estimators remain unbiased for their associated causal estimands \parencite{bojinov2019time}. The distinction is that our design objective in the following section is derived from the finite-order linear model. Therefore, the theoretical precision gains of the optimal designs are guaranteed for the model-based objective, but they do not automatically extend to alternative design-based estimands.
\end{remark}

\subsection{Design Objective and Optimization}
Under the linear impulse-response model \eqref{eq:main_model}, it is easy to verify that 
\[
\widehat{\beta} = \beta + \Bigl( \frac1{T_\text{eff}}\sum_{t=K}^T (X_t-\bar X)(X_t-\bar X)^\top \Bigr)^{-1} \Bigl( \frac1{T_\text{eff}}\sum_{t=K}^T (X_t-\bar X)\epsilon_t \Bigr)\;.
\]
Therefore, 
\[
\E (\widehat{\beta} \mid \{Z_t\}_{t=1}^T) = \beta\;, \quad \Cov(\widehat{\beta}\mid \{Z_t\}_{t=1}^T) \propto \Sigma^{-1}\;, \quad \Sigma = \frac1{T_\text{eff}}\sum_{t=K}^T (X_t-\bar X)(X_t-\bar X)^\top\;.
\]
That is, conditional on the realized treatment assignments, the OLS estimate is unbiased and its variance is governed by the inverse of the Gram matrix $\Sigma$, which is the scaled Fisher information for $\beta$ after partialling out the intercept. Consequently, $\widehat{\tau}_w$ is an unbiased estimator of $\tau_w$. By the law of total variance, $\widehat{\tau}_w$ has the variance as below:
\[
\var(\widehat{\tau}_w) \propto w^\top \E(\Sigma^{-1}) w\;.
\]
This leads to the definition of the design objective function.
\begin{definition}\label{def:obj}
For a given vector $w\in \R^{K}$, we define the asymptotic design objective as
\begin{equation}\label{eq:obj}
\cL_{asy}(\{(\rho_t, \gamma_t)\}_{t=1}^{T-1}; w) = w^\top (\E \Sigma)^{-1} w\;.
\end{equation}
\end{definition}
Our objective $\cL_{asy}(\cdot)$ can be interpreted as an asymptotic measure of the estimation variance $w^\top \E(\Sigma^{-1}) w$. To see this, suppose that the matrix $\Sigma$ converges in probability to a limit $\Sigma_* \in \R^{K\times K}$, a positive definite matrix. 
Then, under suitable inverse-moment conditions and the continuous mapping theorem, we have
\[
\lim_{T\to\infty} (\E \Sigma)^{-1} =
\lim_{T\to\infty}
\E (\Sigma^{-1}) = \Sigma_*^{-1}\;.
\]
Thus, when the normalized Gram matrix converges to a positive definite limit, $w^\top \E(\Sigma^{-1}) w$ and $w^\top (\E\Sigma)^{-1} w$ are asymptotically equivalent. 

Notably, the matrix $\E\Sigma$ in Equation~\eqref{eq:obj} may not always be invertible, in which case the design objective is ill-defined. We handle such singular cases using the following convention:
\begin{equation}\label{eq:inv}
w^\top (\E\Sigma)^{-1} w
\coloneqq 
\begin{cases}
w^\top (\E\Sigma)^\dagger w, & \text{if } w\in\operatorname{Im}(\E \Sigma)\\
\infty, & \text{if } w\notin\operatorname{Im}(\E \Sigma)
\end{cases}\;,
\end{equation}
where $A^\dagger$ denotes the Moore-Penrose inverse of a matrix $A$. In other words, for $w\notin\operatorname{Im}(\E\Sigma)$, the corresponding treatment effect is not estimable under the design, and we set the objective value to $\infty$.

The objective \(\cL_{asy}\) therefore provides a surrogate loss for measuring estimation variance. Its main advantage is computational tractability: it only requires inverting the expected Gram matrix, \((\E\Sigma)^{-1}\), rather than computing the expectation of the inverse, \(\E(\Sigma^{-1})\), as in the exact variance expression. Moreover, combined with the moment formula in Proposition \ref{prop:moments}, the asymptotic objective yields theoretical insights into the optimal design, as discussed later in Sections~\ref{sec:homog} and~\ref{sec:exact}. In finite samples, \(\cL_{asy}\) can also be interpreted as an optimistic objective: since the map \(\Sigma\mapsto w^\top\Sigma^{-1}w\) is convex on the positive definite cone, Jensen's inequality gives
\[
\cL_{asy} = w^\top (\E \Sigma)^{-1} w \le w^\top \E (\Sigma^{-1}) w\;.
\]

Suppose that we are given a finite set of $n$ candidate designs $\cD = \{d_i\}_{i=1}^n$, where each design $d_i = \{(\rho_t^i,\gamma_t^i)\}_{t=1}^{T-1}$ is constructed using, for instance, grid search or domain knowledge. Also, we consider a finite set of estimands of interest $\cW\subset\R^K$. Under the design objective $\cL_{asy}$, we propose the design optimization procedure in Algorithm \ref{alg:design_opt}. 

\begin{algorithm}[h]
\caption{Markovian Design Optimization}
\label{alg:design_opt}
\begin{algorithmic}[1]
\Require Candidate design class \(\cD=\{d_i\}_{i=1}^n\), estimand class \(\cW\)

\For{\(i=1,\dots,n\)}
    \State Compute the expected Gram matrix \(\E\Sigma^i\) under design $d_i = \{(\rho_t^i,\gamma_t^i)\}_{t=1}^{T-1}$.
    \State Evaluate the worst-case asymptotic objective
    \[
    V_i
    \coloneqq
    \max_{w\in\cW}
    \cL_{asy}(d_i;w)
    =
    \max_{w\in\cW}
    w^\top(\E\Sigma^i)^{-1}w.
    \]
\EndFor

\State Choose
\[
i^\star \in \arg\min_{i=1,\dots,n} V_i.
\]
\State \Return \(d_{i^\star}\).
\end{algorithmic}
\end{algorithm}

Algorithm~\ref{alg:design_opt} formalizes the high-level idea illustrated in Figure~\ref{fig:demo}. When targeting a single estimand, the estimand class is the singleton $\cW=\{w\}$ and the optimization reduces to a targeted design optimization: 
\begin{equation}\label{eq:target}
\text{Targeted Design Optimization}: \min_{i=1, \dots, n} w^\top (\E \Sigma^i)^{-1} w\;.
\end{equation}
This corresponds to selecting the optimal design for each specific estimand shown in the columns of Figure~\ref{fig:demo}. When targeting a class of estimands, we seek a design that is robust for estimating all effects in $\cW$: 
\begin{equation}\label{eq:robust}
\text{Robust Design Optimization}: \min_{i=1, \dots, n} \max_{w\in \cW} w^\top (\E \Sigma^i)^{-1} w\;.
\end{equation}
This yields a minimax design procedure for N-of-1 trials \parencite{basse2023minimax,cox2000theory}, corresponding to the design selected under the worst-case criterion in the rightmost column of Figure~\ref{fig:demo}. Importantly, the algorithm depends only on the design parameters and the estimand class, and can therefore be implemented before the experiment is conducted. Computationally, the key challenge lies in Step 2 of evaluating the expected Gram matrix $\E \Sigma^i$. This can be done using Monte Carlo approximation or, alternatively, by deriving closed-form expressions through moment calculations (Proposition~\ref{prop:moments}). We develop the latter approach in detail in the next section.

While Algorithm \ref{alg:design_opt} provides a general procedure for selecting an optimal design, it is formulated for a specified set of candidate designs and target estimands. As a result, it offers limited insight on the temporal treatment dependence of an optimal design. In the following sections, we theoretically characterize the optimal temporal structure through two important design parameters: the switching probability at each time and the lengths of persistent treatment/control periods. To understand how these features should be optimized, it is both practically and theoretically useful to study structured design classes over a continuum of design parameters or estimands. In Sections~\ref{sec:homog} and~\ref{sec:exact}, we study two such classes, random-switch designs and cycle-switch designs, and use them to theoretically characterize optimal design patterns.

\section{Optimal Design Theory for Switching Probabilities}\label{sec:homog}
In this section, we establish theoretical results on optimal temporal dependence in $\{Z_t\}_{t=1}^T$ under the large-$T$ regime. We focus on random-switch designs, a subclass of Markovian assignment designs.
\begin{definition}\label{def:soft}
A random-switch design with parameters $(\rho,\gamma)$ is a Markovian assignment design satisfying
\[
\Pr(Z_1 = 1) = \frac{\rho}{\rho+\gamma}\;, \quad
P_t = P = 
\begin{bmatrix}
1-\rho & \rho \\
\gamma & 1-\gamma
\end{bmatrix}
\quad \text{for } t=1,\dots,T-1\;.\footnote{In the case $\rho = \gamma = 0$, we manually define $\frac{\rho}{\rho + \gamma} = 1/2$.}
\]
\end{definition}
By focusing on random-switch designs, designing temporal dependence reduces to selecting the switching probabilities $\rho$ and $\gamma$. Although such designs form only a subclass of Markovian assignment designs, they provide theoretical insight into how the switching probability can be optimized over the experimental horizon. Moreover, the random-switch design reduces to i.i.d. Bernoulli(1/2) randomization when $\rho = \gamma = 1/2$ (Example \ref{ex:indep}). Thus, optimizing over this design class allows us to assess potential precision gains relative to the Bernoulli benchmark.

In Definition \ref{def:soft}, we set the initial treatment probability to $\rho/(\rho+\gamma)$ which corresponds to the stationary distribution $(\gamma/(\rho+\gamma),\rho/(\rho+\gamma))$ of the homogeneous Markov transition matrix $P$. This ensures that the Markov chain $\{Z_t\}_{t=1}^T$ is stationary and simplifies the analysis. 
In the following, we analyze the behavior of the optimal design parameters $\rho^\star$ and $\gamma^\star$ for the targeted design optimization \eqref{eq:target} and the robust design optimization \eqref{eq:robust}, over a continuum of design parameters.

\subsection{Targeted Design Optimization}
First, we consider the design optimization by targeting a treatment effect under a specified estimand vector $w$. We solve the targeted design optimization \eqref{eq:target} within random-switch designs:  
\begin{equation}\label{eq:target_soft}
\min_{\rho, \gamma\in [\delta, 1 - \delta]} w^\top (\E\Sigma)^{-1} w\;,
\end{equation}
where $w$ is pre-specified to target a particular estimand. For technical reasons that will be discussed later, we exclude estimand vectors satisfying either $w_k = w_{k+1}$ for all $k=1,\dots,K-1$ or $w_k = -w_{k+1}$ for all $k=1,\dots,K-1$. These vectors lead to boundary optima and include important estimands such as the cumulative treatment effect $(1,\dots,1)$. Accordingly, we restrict attention to the compact class of candidate designs with $\rho, \gamma \in [\delta, 1 - \delta]$. We handle the boundary cases later in Section \ref{sec:exact} using cycle-switch designs. The following theorem characterizes the asymptotic solution behavior as $T$ goes to infinity. The proof can be found in Section \ref{sec:proof_homog} of the Appendix. 

\begin{theorem}\label{thm:target_random}
Suppose that $K$ is fixed with $K\ge 2$. Consider an estimand vector $w$ that satisfies neither $w_k = w_{k+1}$ nor $w_k = - w_{k+1}$ for all $k = 1, \dots, K-1$. Then, there exists a unique solution $(\rho^\star, \gamma^\star)$ in $(0, 1)^2$ satisfying
\begin{gather}
\rho = \gamma\;, \label{eq:alpha}\\
4 \rho^2 \sum_{k=1}^{K-1} w_k w_{k+1} + (2\rho-1) \sum_{k=1}^{K-1}(w_k-w_{k+1})^2 = 0\;. \label{eq:rho}
\end{gather}
Set $\delta\in(0, 1/2)$ to be a sufficiently small constant such that $(\rho^\star, \gamma^\star) \in (\delta, 1 - \delta)^2$, and let $(\rho_T^\star, \gamma_T^\star)$ be any minimizer of the targeted design optimization \eqref{eq:target_soft} under time $T$. Then we have
\[
\lim_{T\to\infty} (\rho_T^\star, \gamma_T^\star) = (\rho^\star, \gamma^\star)\;.
\]
\end{theorem}

Based on Theorem~\ref{thm:target_random}, the optimal design exhibits two limiting features:
\begin{itemize}
  \item \textit{Stationary treatment probability of $1/2$.} According to Equation~\eqref{eq:alpha}, the marginal stationary treatment probability 
  $\frac{\rho}{\rho + \gamma}$ converges to $1/2$. This is consistent with Theorem 1 of \textcite{bojinov2023design}, which establishes the optimality of fair coin flipping. Our result recovers the same optimal marginal treatment probability while allowing for temporal dependence among treatment assignments.
  \item \textit{Estimand-specific switching probability.} According to Equation \eqref{eq:rho}, the optimal switching probability $\rho^\star$ is determined by a quadratic equation involving the estimand vector $w$. Thus, the optimal degree of temporal treatment dependence is determined by the target estimand.
\end{itemize}

In the examples below, we apply Theorem \ref{thm:target_random} to different choices of $w$ and derive the limiting optimal designs.
\begin{example}\label{ex:lag}[Lag-specific treatment effect]
For $w  = e_k$ for any $k = 1, \dots, K$, $\rho^\star$ solves the equation $2 \rho - 1 = 0$.
Therefore, the limiting optimal design is $\rho^\star = \gamma^\star = 1/2$, corresponding to the i.i.d. Bernoulli(1/2) randomization.
\end{example}

\begin{example}\label{ex:contrast}[Contemporaneous vs. lagged treatment effects]
For $w_1 = 1$ and $w_k = -1$ for $k = 2, \dots, K$, $\rho^\star$ solves the equation $(K-3) \rho^2 + 2 \rho - 1  = 0$.
From the quadratic equation, the limiting optimal design is $\rho^\star = \gamma^\star = {1}/ (1+\sqrt{K-2})$. Figure \ref{fig:loss_glob} validates the convergence of the optimal design parameters under $K = 5$. 
\end{example}
\begin{figure}[th!]
\centering
\subfloat[$T = 10$]{\includegraphics[width = 0.33\linewidth]{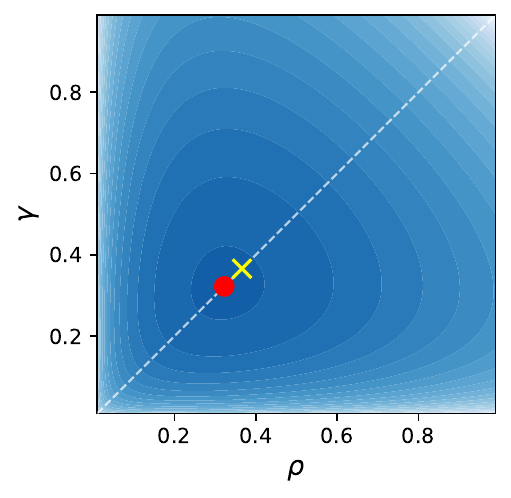}} 
\subfloat[$T = 20$]{\includegraphics[width = 0.33\linewidth]{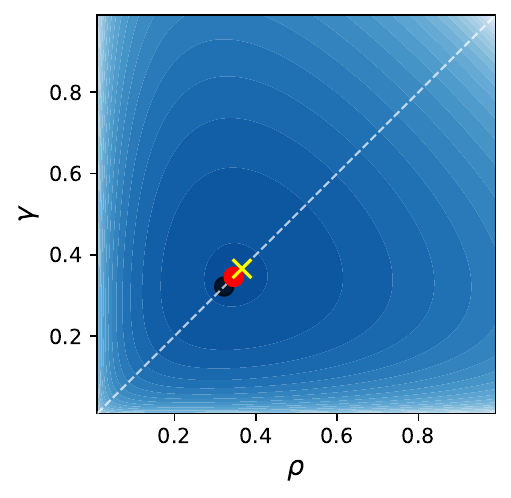}} 
\subfloat[$T = 100$]{\includegraphics[width = 0.33\linewidth]{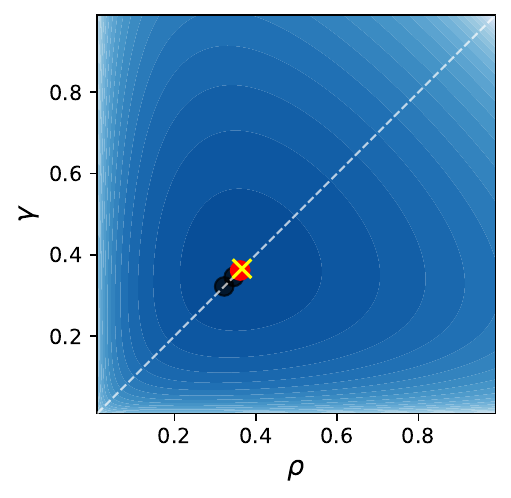}}
\caption{Loss landscape of $\cL_{asy}(\rho,\gamma; w)$ for $w=(1,-1,-1,-1,-1)$. For each $T$, the optimal solution $(\rho_T^\star, \gamma_T^\star)$ is highlighted in red. The limiting solution ($\rho^\star = \gamma^\star \approx 0.366$ from Example \ref{ex:contrast}) is marked by the yellow cross, and the trajectory of finite-sample solutions is shown in black.}
\label{fig:loss_glob}
\end{figure}

Theorem \ref{thm:target_random} excludes two boundary cases: $w_k = w_{k+1}$ for all $k=1,\dots,K-1$, and $w_k = -w_{k+1}$ for all $k=1,\dots,K-1$. For instance, for the cumulative treatment effect with $w_k = 1$, a heuristic application of Theorem \ref{thm:target_random} yields $\rho^\star = \gamma^\star = 0$. Intuitively, because all lag coefficients enter with the same sign, the design favors long uninterrupted treatment and control periods rather than frequent switching. In the large-$T$ limit, this pushes the switching probability to zero and leads to perfect temporal dependence. However, this limiting design can be practically infeasible: when switches become too rare, the realized assignment path may contain no switches, making the OLS estimator ill-defined. In Section~\ref{sec:exact}, we introduce cycle-switch designs, which explicitly control the lengths of treatment and control periods and yield feasible, powerful designs for these boundary cases.

\subsection{Robust Design Optimization}
Next, we consider the robust design optimization problem
\begin{equation}\label{eq:robust_soft}
  \min_{\rho, \gamma\in [\delta, 1 - \delta]} \max_{\|w\|\le 1}  w^\top (\E \Sigma)^{-1} w \;.
\end{equation}
The criterion in \eqref{eq:robust_soft} is robust with respect to the estimand direction $w$: it seeks a design that performs well uniformly over all normalized linear combinations of the lag-specific treatment effects. Thus, rather than targeting a particular treatment effect, the robust design protects against the least precisely estimated direction in the parameter space. This criterion is closely related to E-optimality \parencite{cox2000theory,pukelsheim2006optimal} in the theory of optimal experimental design.
\begin{remark}[Connection to E-optimality]
Let $\lambda_{\max}(\cdot)$ and $\lambda_{\min}(\cdot)$ denote the largest and
smallest eigenvalues of a matrix, respectively. Since
\begin{align*}
\max_{\|w\|\le 1} w^\top (\E\Sigma)^{-1} w
=
\|(\E \Sigma)^{-1}\|
=
\frac{1}{\lambda_{\min}(\E \Sigma)}\;,
\end{align*}
the robust design optimization \eqref{eq:robust_soft} is equivalent to $\max_{\rho,\gamma\in[\delta,1 - \delta]} \lambda_{\min}(\E \Sigma)$. Therefore, our design objective coincides with the E-optimality criterion in optimal experimental design \parencite{cox2000theory,pukelsheim2006optimal}. The distinction is that we study this criterion in N-of-1 experiments whose assignment process follows random-switch designs.
\end{remark}

The following theorem characterizes the asymptotic solution in the large-$T$ regime. Its proof can be found in Section \ref{sec:proof_homog} of the Appendix.
\begin{theorem}\label{thm:robust_random}
Suppose that $K, \delta$ are fixed with $K \ge 2$ and $\delta\in(0, 1/2)$. Let $(\rho_T^\star, \gamma_T^\star)$ be any minimizer of the robust design optimization \eqref{eq:robust_soft} under time $T$. We have
\[
\lim_{T\to\infty} (\rho_T^\star, \gamma_T^\star) = \Bigl(\frac{1}{2}, \frac{1}{2}\Bigr)\;.
\]
\end{theorem}

By the definition of random-switch designs, the limiting solution $\rho^\star=\gamma^\star=1/2$ corresponds to i.i.d. Bernoulli$(1/2)$ randomization (Example \ref{ex:indep}). Therefore, Theorem~\ref{thm:robust_random} provides a formal justification for the optimality of i.i.d. Bernoulli randomization across normalized linear combinations of lag-specific effects in N-of-1 trials. This conclusion is consistent with the observation of \textcite{liang2023randomization} and with classical robustness results in cross-sectional settings \parencite{wu1981robustness,harshaw2024balancing}. Under the robust criterion above, introducing temporal dependence does not improve estimation precision.

\section{Optimal Design Theory for Treatment Cycles}\label{sec:exact}
As discussed in the previous section, the random-switch analysis has two limitations. First, its theoretical characterization is obtained under interior regularity conditions and excludes important boundary estimands, such as the cumulative effect. These boundary cases are difficult because the estimand vector exhibits strong persistence; for example, $w=(1,\dots,1)$ satisfies $w_{k+1}=w_k$ for all $k$. This is analogous to the unit root problem in time-series analysis, where persistence becomes so strong that the behavior of estimators can differ qualitatively from that in regular setups. Second, an optimal switching probability does not immediately translate into optimal lengths of treatment and control periods in the realized assignment path. Motivated by these limitations, this section turns to cycle-switch designs, which specify the treatment and control cycle lengths directly and thereby provide an exact way to study optimal design in boundary regimes.

Concretely, we consider extreme Markov transition matrices that directly incorporate periodic switches in the design process. Recall that $P_t$ is the Markov transition from $t$ to $t+1$, determined by parameters $\rho_t$ and $\gamma_t$.
\begin{definition}[Cycle-switch design]
Let $t_m = 1+ml$ for $m=0,\dots,M$ with $M=\lfloor \frac{T-1}{l}\rfloor$. The cycle-switch design is given by $Z_1=1$, and for $t=1,\dots,T-1$,
\[
P_t =
\begin{bmatrix}
    1 & 0 \\
    0 & 1
\end{bmatrix}
\qquad
\text{if } t\notin\{t_1-1,\ldots,t_M-1\}\;,\qquad P_{t_m-1} =
\begin{bmatrix}
    0 & 1 \\
    1 & 0
\end{bmatrix}\;, \qquad m=1,\ldots,M\;.
\]
\end{definition}

Intuitively, the cycle-switch design generates the deterministic treatment sequence
\[
(Z_1,\dots,Z_T)
=
(\underbrace{1,\dots,1}_{l},
\underbrace{0,\dots,0}_{l},
\underbrace{1,\dots,1}_{l},
\underbrace{0,\dots,0}_{l},
\dots)\;.
\]
Under cycle-switch designs, the design of temporal treatment dependence boils down to the selection of the design parameter $l$. The length $l$ of treatment and control periods plays a role analogous to $(\rho,\gamma)$ in the random-switch designs. The cycle-switch design may also be viewed as a deterministic counterpart of regular switchback experiments (Example~\ref{ex:regular}), where the cycle length is fixed and a switch is enforced at the beginning of each new cycle.

Since the design is deterministic, the design objective in Definition \ref{def:obj} satisfies
\[
\cL_{asy}(\{(\rho_t, \gamma_t)\}_{t=1}^{T-1}; w) = w^\top (\E\Sigma)^{-1} w = w^\top \Sigma^{-1} w = w^\top \E(\Sigma^{-1}) w\;.
\]
Therefore, the asymptotic design objective is equal to the exact estimation variance up to a constant multiplier.

The matrix $\Sigma$ can be precisely characterized under cycle-switch designs, which serves as the building block for the theoretical analysis in this section. To simplify notation, define
\begin{equation}\label{eq:fq}
f(d) = \frac{|d\bmod 2l - l|}{2l}\;,
\end{equation}
where $a \bmod b$ denotes the least nonnegative residue after dividing $a$ by $b$. 
The following proposition gives an explicit expression for $\Sigma$. The proof can be found in Section \ref{sec:proof_exact} of the Appendix.

\begin{proposition}\label{prop:sigma_hard}
Suppose that $T_\text{eff}$ is divisible by $2l$. Then $\Sigma_{ij} = f(i-j) - \frac{1}{4}$.
\end{proposition}

Motivated by Proposition~\ref{prop:sigma_hard}, we define the following matrix for a general block length $l$ 
\begin{equation}\label{eq:tilde_sigma}
\widetilde{\Sigma}_{ij} = f(i-j) - \frac{1}{4}\;,
\end{equation}
and characterize the optimal design parameter under the corresponding design objective 
\[
w^\top \widetilde{\Sigma}^{-1} w\;.
\]
The matrix $\widetilde{\Sigma}$ coincides exactly with $\Sigma$ in the divisible case and serves as an approximation in nondivisible cases. In fact, as shown in Section~\ref{sec:proof_exact} of the Appendix, $\widetilde{\Sigma}$ is the limit of $\Sigma$ as $T\to\infty$.

Figure \ref{fig:sigma} visualizes the matrix $\widetilde{\Sigma}$ under different values of $K$ and $l$. The matrix exhibits a cyclic structure induced by the cyclic structure of the treatment assignments. Moreover, Figure~\ref{fig:sigma}(a) shows that $\widetilde{\Sigma}$ can be low rank and therefore noninvertible. In such settings, we apply the convention in Equation \eqref{eq:inv} and set the objective value to $\infty$ for non-identifiable estimands. Similar to Section \ref{sec:homog}, we first study the targeted design optimization \eqref{eq:target} under $\widetilde{\Sigma}$, and then consider the robust design optimization \eqref{eq:robust}.

\begin{figure}
  \centering
  \subfloat[$K = 20, l = 5$]{\includegraphics[width=0.4\linewidth]{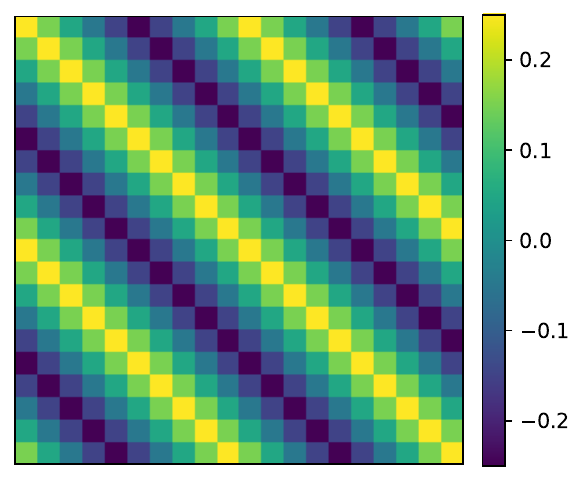}}
  \subfloat[$K = 20, l = 30$]{\includegraphics[width=0.4\linewidth]{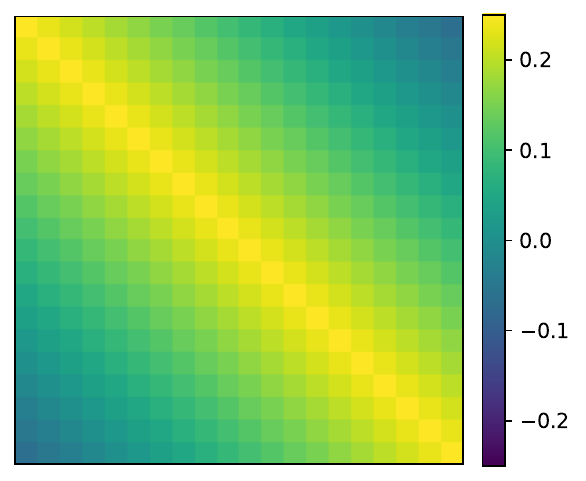}}
  \caption{Visualization of the matrix $\widetilde{\Sigma}$ under cycle-switch designs.}\label{fig:sigma}
\end{figure}

\subsection{Targeted Design Optimization}
In this section, we focus on $\cW = \{w\}$ and cycle-switch designs such that the targeted design optimization \eqref{eq:target} reduces to:
\[
\min_{l\in\bZ_{+}} w^\top \widetilde{\Sigma}^{-1} w\;.
\]
As discussed above, the matrix $\widetilde{\Sigma}$ may be rank-deficient, so our analysis proceeds in two parts. When $l<K$, $\widetilde{\Sigma}$ exhibits a cyclic structure that yields a rank-deficient covariance matrix; we characterize the class of identifiable estimands in this case. When $l\geq K$, $\widetilde{\Sigma}$ is always full rank, and we derive a concrete optimality condition for $l$. 

\paragraph{$l < K$ setups.}
Recall that $e_k$ denotes the $k$-th canonical basis vector in $\mathbb{R}^K$. The following result characterizes the set of identifiable effects. See Section \ref{sec:proof_hard_target} of the Appendix for the proof. 
Recall that, under the convention in Equation \eqref{eq:inv}, $w^\top \widetilde{\Sigma}^{-1} w < \infty$ implies that $w$ is estimable.

\begin{theorem}\label{thm:identification}
Suppose that $K > l$. Then $w^\top \widetilde{\Sigma}^{-1} w < \infty$ if and only if $w$ is a linear combination of $b_u$, $u = 1, \dots, l$, where 
\[
b_u = \sum_{r=0}^{\lfloor (K-u)/l\rfloor}(-1)^r e_{u+rl}\;.
\]
\end{theorem}
Theorem~\ref{thm:identification} shows that when the switching cycle $l$ is shorter than the lag horizon $K$, the design can only distinguish alternating contrasts across lag positions separated by $l$. Treatment effects that move together within these cyclic groups are not separately recoverable. Concretely, the example below characterizes the identifiable estimands in a simple setting.

\begin{example}
For $l=2$ and $K=4$, an estimand vector $w$ is identifiable if and only if there exist constants $a_1, a_2 \in \R$ such that
\[
w = a_1(1,0,-1,0)^\top + a_2(0,1,0,-1)^\top\;.
\]
Thus, the design can identify contrasts between treatment effects at lag positions separated by the block length $l=2$, such as the difference between $\beta_1$ and $\beta_3$ or between $\beta_2$ and $\beta_4$. By contrast, estimands with homogeneous structure across these paired lag positions, such as $(1,0,1,0)^\top$ or $(0,1,0,1)^\top$, are not identifiable under this design.
\end{example}

\paragraph{$l \ge K$ setups.}
When $l\ge K$, the design objective admits the following explicit characterization. The proof can be found in Section \ref{sec:proof_hard_target} of the Appendix.
\begin{theorem}\label{thm:exact_sol2}
Suppose that $K\ge 2$ and $l\ge K$. Then 
\begin{equation}\label{eq:exact_opt}
w^\top \widetilde{\Sigma}^{-1} w = l A(w) + \frac{l}{l-K+1} B(w)\;,
\end{equation}
where
\[
A(w) = \sum_{k=1}^{K-1} (w_k-w_{k+1})^2\;, \qquad
B(w) = (w_1+w_K)^2\;.
\]
\end{theorem}
Intuitively, Theorem~\ref{thm:exact_sol2} decomposes the design objective into an interior variation term, $A(w)$, and a boundary term, $B(w)$. The optimal block length is determined by $w$ through $A(w)$ and $B(w)$. Combining Theorems \ref{thm:identification} and \ref{thm:exact_sol2}, we characterize the optimal design through the block length $l$ for cumulative and lag-specific treatment effects. The optimal block length $l^*$ below is validated numerically in Figure \ref{fig:loss_hard_switch}.

\begin{corollary}[Truncated cumulative treatment effect]\label{cor:cumulative}
For any cumulative effect satisfying $w_1 = \dots = w_m = 1$ and $w_{m+1} = \dots = w_K = 0$ with $1 < m < K$, we have
\[
w^\top \widetilde{\Sigma}^{-1} w = 
\begin{cases}
\infty, & l<K\;, \\[1ex]
l+\dfrac{l}{l-K+1}, & l\ge K\;.
\end{cases}
\]
\end{corollary}
In the corollary above, we focus on a truncated cumulative effect to simplify the discussion. The full cumulative effect corresponding to $w=(1,\dots, 1)$ can be handled similarly, up to some additional technical nuances. The first $l < K$ case follows from Theorem \ref{thm:identification}, and the second $l\ge K$ case is a result of Theorem \ref{thm:exact_sol2}. To solve the optimization problem, let $l = K-1 + h$. Then,
\[
w^\top \widetilde{\Sigma}^{-1} w = K + h + \frac{K - 1}{h}\;.
\]
The choice of window length therefore reflects a tradeoff. 
Balancing the two terms involving $h$ gives $h^\star \approx \sqrt{K-1}$, and hence $l^\star \approx K-1+\sqrt{K-1}$.

\begin{corollary}[Lag-specific treatment effect]\label{cor:lag}
For a lag-specific treatment effect $w = e_k$ with $k = 2, \dots, K-1$, we have
\[
w^\top \widetilde{\Sigma}^{-1} w = 
\begin{cases}
\infty, & l<\max\{k, K-k+1\}\;, \\[1ex]
2l, & l\ge \max\{k, K-k+1\}\;.
\end{cases}
\]
\end{corollary}
Similarly, Corollary~\ref{cor:lag} combines Theorems \ref{thm:identification} and \ref{thm:exact_sol2}, providing a complete characterization over all values of $l$. Intuitively, the choice of $l$ reflects a tradeoff between identifiability and variance. If $l<\max\{k,K-k+1\}$, then the lag-specific estimand $w=e_k$ is not identifiable, because the corresponding lag effect cannot be separated from other lagged treatment effects induced by the cycle-switch structure. Among identifiable designs, smaller values of $l$ reduce the variance from the measurement errors $\epsilon_t$. Therefore, for estimating the lag-specific effect given $w=e_k$, the optimal design is $l^\star = \max\{k,K-k+1\}$. More generally, if the goal is to minimize the worst-case variance over interior lag-specific effects $w=e_2,\dots,e_{K-1}$, corresponding to lags $1,\dots,K-2$, the optimal choice is $l^\star=K-1$.

\begin{figure}[th!]
\centering
\subfloat[Cumulative effect]{\includegraphics[width = 0.5\linewidth]{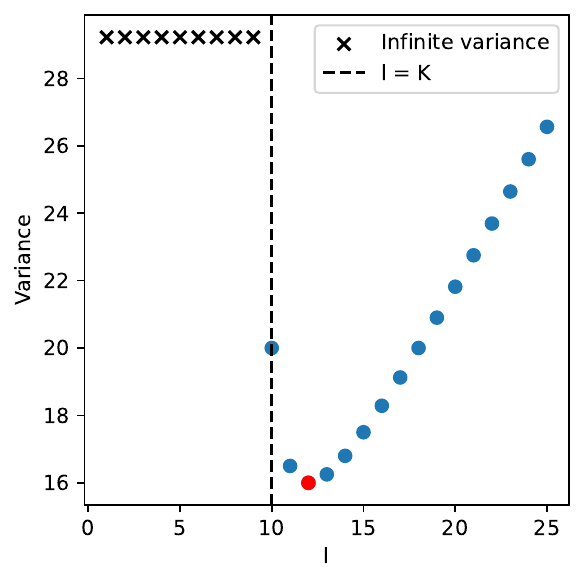}} 
\subfloat[Lag-specific effect]{\includegraphics[width = 0.5\linewidth]{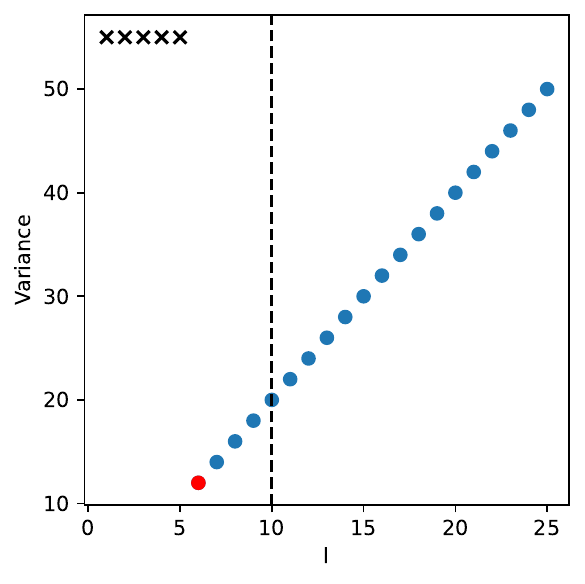}}
\caption{Variance curve of cycle-switch designs for $w = (1, \dots, 1, 0)$ (cumulative effect) and $w = (0,0,0,0,1,0,\dots, 0)$ (lag-specific effect $\beta_5$) with $K = 10$. The optimal solutions in (a) and (b) are highlighted in red, which validate Corollaries \ref{cor:cumulative} and \ref{cor:lag}, respectively.}
\label{fig:loss_hard_switch}
\end{figure}

\subsection{Robust Design Optimization}
We consider the robust design optimization problem over $\|w\|\leq 1$:
\begin{equation}\label{eq:robust_hard}
\min_{l \in \bZ_+} \max_{\|w\|\le 1} w^\top \widetilde{\Sigma}^{-1} w\;.
\end{equation}
When $\widetilde{\Sigma}$ is singular, we define the worst-case variance over $\|w\|\le 1$ to be $\infty$, since some estimand directions are not identifiable. The following theorem shows that the optimal cycle-switch design is $l^\star=K$. See Section \ref{sec:proof_hard_robust} of the Appendix for the proof.

\begin{theorem}\label{thm:robust2}
Suppose that $K$ is even. The optimizer of the robust design optimization \eqref{eq:robust_hard} is $l^\star = K$.
\end{theorem}

The proof of Theorem~\ref{thm:robust2} rests on two key observations. First, if $l<K$, then $\widetilde{\Sigma}$ is rank-deficient, so the worst-case variance is infinite. Thus, block lengths shorter than $K$ are not robust enough to detect all possible estimand directions. Second, if $l\geq K$, then the worst-case variance grows linearly in $l$, yielding the corner solution $l^\star=K$. Intuitively, the optimal switching cycle should match the length of the treatment history that may influence the outcome.

\section{Simulations}\label{sec:simu}
In this section, we evaluate the performance of the random- and cycle-switch designs proposed in this paper and compare them with commonly used alternatives, including i.i.d. Bernoulli randomization and the optimal regular switchback design of \parencite{bojinov2023design}. We first consider data-generating processes that follow Model~\eqref{eq:main_model}, so that the working model is correctly specified. We then conduct a series of robustness checks under model misspecifications that violate Model \eqref{eq:main_model}. Overall, random-switch designs achieve the smallest estimation variance for large $T$ and remain robust across a range of model specifications. Cycle-switch designs attain the smallest estimation variance for estimating the cumulative treatment effects in finite samples, although they are less robust to certain forms of model misspecification.

\subsection{Design Performance under Different Objectives}
We consider Model~\eqref{eq:main_model} with $K=5$ and evaluate the worst-case estimation variance over several classes of estimands. We consider $T\in\{20,30,\dots,100\}$. For each value of $T$, we solve for the optimal random-switch design using Algorithm \ref{alg:design_opt}, with $(\rho,\gamma)$ selected from a $500\times 500$ grid on $[0,1]^2$. Similarly, we solve for the optimal cycle-switch design using Algorithm \ref{alg:design_opt}, with $l$ selected from all integers from $1$ to $T$.

We compare the optimal random- and cycle-switch designs with i.i.d. Bernoulli(1/2) randomization and the optimal regular switchback design (ORSB) of \textcite{bojinov2023design}, implemented with the lag horizon $K=5$. For each candidate design and each $w\in\cW$, we evaluate the estimation error $w^\top \E(\Sigma^{-1}) w / T_\text{eff}$ by Monte Carlo over 100 realized treatment paths, and then take the worst case over $w\in\cW$.

We first focus on $\cW=\{e_1,\dots,e_5\}$, so that the objective measures the worst-case variance for estimating lag-specific treatment effects. Figure~\ref{fig:compare}(a) shows that the random-switch design and i.i.d. Bernoulli randomization achieve the smallest worst-case variance across the range of $T$. In fact, the optimal random-switch design is approximately $(\rho^\star,\gamma^\star)=(0.5,0.5)$, which recovers i.i.d. Bernoulli randomization. This is consistent with the theory in Section~\ref{sec:homog} (Example \ref{ex:lag}). The cycle-switch design yields a larger worst-case variance, while ORSB performs substantially worse in this setting, reflecting that it is optimized for carryover-type effects rather than lag-specific treatment effects \parencite{bojinov2023design}.

In Figure~\ref{fig:compare}(b), we consider $\cW=\{\sum_{i=1}^k e_i \mid k=1,\dots,K\}$, so that the objective measures the worst-case variance for estimating cumulative effects over varying lag horizons. We observe that the cycle-switch design achieves the smallest worst-case variance, followed by the random-switch design and ORSB. This finding suggests that cycle-switch designs are particularly well suited for estimating cumulative effects. Notably, i.i.d. Bernoulli randomization exhibits the largest variance in this setting, reflecting the tradeoff between robustness and accuracy.

\begin{figure}[h]
  \centering
  \subfloat[Lag-specific effect]{\includegraphics[width=0.4\linewidth]{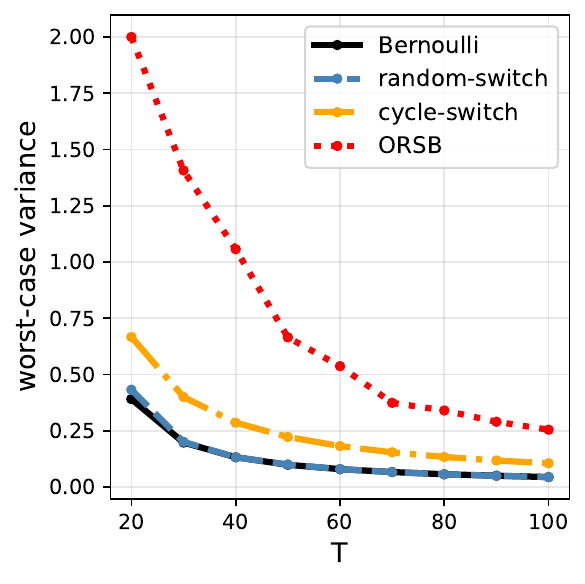}}
  \subfloat[Cumulative effect]{\includegraphics[width=0.4\linewidth]{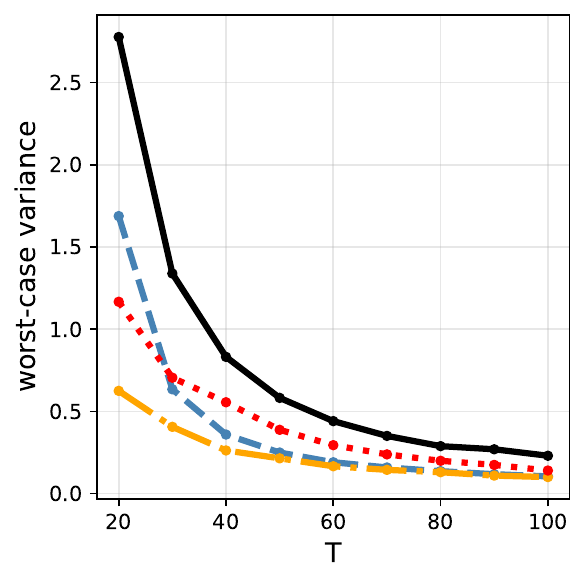}}
  \caption{Worst-case variances for different designs. In (a) and (b), we target the lag-specific and cumulative effects, respectively. In (a), for sufficiently large $T$, the optimal designs satisfy $(\rho^\star, \gamma^\star) \approx (0.5, 0.5)$ in the random-switch setting, and $l^\star = K = 5$ in the cycle-switch setting. In (b), for sufficiently large $T$, the optimal designs satisfy $(\rho^\star, \gamma^\star) \approx (0, 0)$ in the random-switch setting, and $l^\star = 7$ in the cycle-switch setting.}\label{fig:compare}
\end{figure}

\subsection{Misspecification}
We assess the robustness of different designs when the model is misspecified. In the estimation, we focus on $K = 3$ and the OLS estimate $w^\top \widehat{\beta}$ with $w= (1, 1, 1)$. That is, the target effect is the cumulative treatment effect. 

First, we focus on a linear additive outcome-generating process:
\[
Y_t  = \alpha + \log(t) + \sum_{k=1}^K \beta_k Z_{t-k+1} + \epsilon_t\;,
\qquad
\epsilon_t \stackrel{\mathrm{iid}}{\sim} \mathcal{N}(0,1)\;.
\]
Under the data-generating process above, our model is misspecified as the time trend $\log(t)$ is missing in Equation \eqref{eq:main_model}. We focus on $\beta_1 = \beta_2 = \beta_3 = 1$ and the target estimand satisfies $\tau = \sum_{k=1}^3 \beta_k = 3$.

We implement the designs as in the previous section and report the bias, variance, and mean squared error of the OLS estimator under the misspecified model in Table~\ref{tab:mis}. Panel~A shows that random-switch designs remain relatively robust, with MSE comparable to that of the Bernoulli design. By contrast, cycle-switch designs exhibit substantial bias under this form of misspecification. Thus, when unmodeled time trends may be present in the outcome process, random-switch designs can provide a more robust alternative to cycle-switch designs.

Next, we consider the same linear model as in Equation~\eqref{eq:main_model} with $\beta_1 = \beta_2 = \beta_3 = 1$ (i.e., the target estimand $\tau = 3$), but misspecify the error process. Specifically, we generate errors from a persistent AR(1) process, thereby violating the independence assumption and the variance objective. Panel~B shows that both random-switch and cycle-switch designs remain relatively robust under this form of error misspecification.

Finally, we consider lag misspecification: the true data-generating process follows Equation~\eqref{eq:main_model} with lag length $K=4$ and $\beta_1 = \beta_2 = \beta_3 = \beta_4 = 1$, whereas both the design and the estimator are constructed using $K=3$. In this case, the target cumulative effect $\tau$ is $4$. Panel~C shows that both random-switch and cycle-switch designs again exhibit robustness under this model violation.

\begin{table}[h!]
\centering
\begin{tabular}{lcccc}
\toprule
 & Bernoulli & Random-switch & Cycle-switch & ORSB \\
\midrule
\multicolumn{5}{l}{\textit{Panel A: Time fixed effects}} \\
\midrule
Bias     & -0.080 & 0.092  & -1.262 & 0.111 \\
Variance & 0.394  & 0.348  & 0.034  & 0.280 \\
MSE      & 0.397  & 0.353  & 1.626  & 0.289 \\
\midrule
\multicolumn{5}{l}{\textit{Panel B: Persistent errors}} \\
\midrule
Bias     & -0.024 & 0.047  & -0.001 & 0.050 \\
Variance & 0.301  & 0.113  & 0.129  & 0.171 \\
MSE      & 0.299  & 0.114  & 0.128  & 0.172 \\
\midrule
\multicolumn{5}{l}{\textit{Panel C: Lag misspecification}} \\
\midrule
Bias     & -1.027 & -0.169 & -0.037 & -0.482 \\
Variance & 0.182  & 0.041  & 0.034  & 0.072 \\
MSE      & 1.235  & 0.069  & 0.035  & 0.304 \\
\bottomrule
\end{tabular}
\caption{Bias, variance, and MSE for different designs under model misspecification.}\label{tab:mis}
\end{table}

\section{Conclusions}
In this paper, we study optimal experimental design for N-of-1 trials under linear impulse-response models. We introduce a broad class of designs induced by time-inhomogeneous Markov chains and propose a concrete procedure, Algorithm~\ref{alg:design_opt}, for selecting an optimal design. To obtain theoretical insight into temporal dependence in treatment assignments, we focus on two subclasses: random-switch designs and cycle-switch designs. For each class, we characterize the optimal design parameters and show how the target estimand shapes the preferred assignment structure. We evaluate the proposed minimax designs through simulations that reflect a range of practically relevant scenarios.

One promising direction for future work is to extend the framework to panel experiments, where treatments may be designed, potentially heterogeneously, across multiple experimental units. Studying optimal design in such panel settings would be of substantial interest.

\section*{Acknowledgements}
TL gratefully acknowledges support from the NSF Career Award (DMS-2042473)
and the Wallman Society of Fellows at the University of Chicago.

\ifbiblatex
\else
  \bibliographystyle{abbrvnat}
  \bibliography{reference}
\fi

\newpage
\appendix

\section{Proofs of Section \ref{sec:param}}\label{sec:proof}
\begin{proof}[(Proof of Proposition \ref{prop:moments})]
First, we prove by induction that for any $s < t$
\begin{equation}\label{eq:recursion}
\Pr(Z_t = 1 \mid Z_s = 1) = A(s, t-1) + \sum_{i=s}^{t-1} A(i+1, t-1) \rho_i\;.
\end{equation}
The summation is equal to zero if $s>t-1$. For $t = s+1$, we have
\[
\Pr(Z_{t} = 1 \mid Z_s = 1) = 1 - \gamma_s = (1 - \gamma_s - \rho_s) + \rho_s = A(s, t-1) + A(s+1, t-1) \rho_s\;,
\]
which satisfies Equation \eqref{eq:recursion}. Suppose Equation \eqref{eq:recursion} holds for a given $t$. Then, for $t+1$, we have
\begin{align*}
\Pr(Z_{t+1} = 1 \mid Z_s = 1) &= \Pr(Z_{t+1}=1\mid Z_t = 0) \Pr(Z_t = 0 \mid Z_s = 1) \\
&+ \Pr(Z_{t+1}=1 \mid Z_t = 1) \Pr(Z_t = 1 \mid Z_s = 1)\\
&= \rho_t \Pr(Z_t = 0 \mid Z_s = 1) + (1 - \gamma_t) \Pr(Z_t = 1 \mid Z_s = 1) \\
&= \rho_t + (1 - \gamma_t - \rho_t) \Pr(Z_t = 1 \mid Z_s = 1)\;.
\end{align*}
Since $\Pr(Z_t = 1 \mid Z_s = 1) = A(s, t-1) + \sum_{i=s}^{t-1} A(i+1, t-1) \rho_i$ according to Equation \eqref{eq:recursion}, we further obtain
\begin{align*}
\Pr(Z_{t+1} = 1 \mid Z_s = 1) &= \rho_t + (1 - \gamma_t - \rho_t) \bigl(A(s, t-1) + \sum_{i=s}^{t-1} A(i+1, t-1) \rho_i\bigr)\\
&= A(s, t) + \sum_{i=s}^{t-1} A(i+1, t) \rho_i + \rho_t\\
&= A(s, t) + \sum_{i=s}^{t} A(i+1, t) \rho_i\;.
\end{align*}
Hence, Equation \eqref{eq:recursion} is proved. Through a similar argument, we have
\begin{equation}\label{eq:recursion2}
\Pr(Z_t = 1 \mid Z_s = 0) = \sum_{i=s}^{t-1} A(i+1, t-1) \rho_i \;.
\end{equation}
In addition, Equation \eqref{eq:recursion} holds under the $s = t$ case, and we omit its proof for simplicity.

Based on recursions \eqref{eq:recursion}, \eqref{eq:recursion2}, we further obtain
\begin{align*}
\Pr(Z_s = 1) &= \Pr(Z_s = 1\mid Z_1 = 0) \Pr(Z_1 = 0) + \Pr(Z_s = 1\mid Z_1 = 1) \Pr(Z_1 = 1)\\
&= (\sum_{i=1}^{s-1} A(i+1, s-1) \rho_i) \Pr(Z_1 = 0) + \Bigl(A(1, s-1) + \sum_{i=1}^{s-1} A(i+1, s-1) \rho_i\Bigr) \Pr(Z_1 = 1)\\
&= \sum_{i=1}^{s-1} A(i+1, s-1) \rho_i + A(1, s-1) \Pr(Z_1 = 1)\;.
\end{align*}
This completes the proof. 
\end{proof}

\section{Proofs of Section \ref{sec:homog}}\label{sec:proof_homog}
To formulate a closed-form design objective, we need to simplify $\E \Sigma$ in $\cL_{asy}$. The following proposition records a useful observation that the entries of $\E \Sigma$ are linear functions of the first and second moments of treatment assignments. 
\begin{proposition}\label{prop:sigma_soft}
For any $i \le j$, we have
\begin{equation*}
\E \Sigma_{ij} = \frac{1}{T_\text{eff}}\sum_{t=K}^{T} \E Z_{t-i+1} Z_{t-j+1} - \frac{1}{T_\text{eff}^2} \sum_{s=K}^{T}\sum_{t=K}^{T} \E Z_{s-i+1} Z_{t-j+1}\;.
\end{equation*}
\end{proposition}

\begin{proof}[(Proof of Proposition \ref{prop:sigma_soft})]
By definition of $\Sigma$, for any $i,j = 1,\dots,K$, we have
\[
\E \Sigma_{ij}
=
\frac{1}{T_\text{eff}} \sum_{t=1}^{T_\text{eff}} \E Z_{t + K - i} Z_{t + K - j}
-
\frac{1}{T_\text{eff}^2}
\sum_{t=1}^{T_\text{eff}}\sum_{s=1}^{T_\text{eff}}
\E Z_{t + K -i} Z_{s + K -j}\;.
\]
Equivalently, it can be expressed as
\begin{equation*}
\E \Sigma_{ij} = \frac{1}{T_\text{eff}}\sum_{t=K}^{T} \E Z_{t-i+1} Z_{t-j+1} - \frac{1}{T_\text{eff}^2} \sum_{s=K}^{T}\sum_{t=K}^{T} \E Z_{s-i+1} Z_{t-j+1}\;.
\end{equation*}
\end{proof}

Proposition \ref{prop:sigma_soft} and Proposition \ref{prop:moments} jointly provide the analytical formula for evaluating the expected Gram matrix $\E \Sigma$. To further simplify Proposition \ref{prop:moments}, under the random-switch design with $\rho_t = \rho$ and $\gamma_t = \gamma$, we write the product kernel $A(\cdot, \cdot)$ in Proposition \ref{prop:moments} as
\[
A(s, t) = (1 - \gamma - \rho)^{t-s+1} \mathbb{I}\{s \le t\} + \mathbb{I}\{s > t\}\;.
\]
We define $\alpha = \rho / (\gamma + \rho)$, which corresponds to the treatment probability under the stationary distribution $(\gamma / (\gamma + \rho), \rho / (\gamma + \rho))$. In addition, let $q = 1 - \gamma - \rho$. By Proposition \ref{prop:moments} and the formula of geometric sums, we have the following lemma.

\begin{lemma}\label{lem:homog_moments}
For any $1 \le s \le t \le T$, a random-switch design with $(\rho, \gamma)$ satisfies that
\begin{align*}
\E Z_s Z_t = \alpha^2 + \alpha (1 - \alpha) q^{t-s}\;.
\end{align*}
\end{lemma}

\begin{proof}[(Proof of Lemma \ref{lem:homog_moments})]
First, we assume $\gamma + \rho > 0$ such that $q = 1 - \gamma - \rho < 1$.
Based on Proposition \ref{prop:moments} and the kernel expression in the homogeneous case, we have that for $s>1$, 
\begin{align*}
  \E Z_s &= A(1, s-1) \Pr(Z_1 = 1)+ \sum_{i=1}^{s-1} A(i+1, s-1) \rho \\
  &= q^{s-1} \Pr(Z_1 = 1) + \sum_{i=1}^{s-1} q^{s - i - 1} \rho\;.
\end{align*}
Since $q = 1 - \rho - \gamma < 1$, by the formula of geometric sums, we obtain
\begin{align*}
  \E Z_s &= q^{s-1} \Pr(Z_1 = 1) + \frac{1 - q^{s-1}}{1 - q} \rho = q^{s-1} \Pr(Z_1 = 1) + (1 - q^{s-1}) \alpha \;.
\end{align*}
For $s = 1$ and $t > s$, we apply Proposition \ref{prop:moments} to obtain
\begin{align*}
  \E Z_1 Z_t &= \Pr(Z_1 = 1) \bigl(A(1, t-1) + \sum_{i=1}^{t-1} A(i+1, t-1) \rho\bigr) \\
  &= \Pr(Z_1 = 1) \bigl(q^{t-1} + \sum_{i=1}^{t-1} q^{t-i-1} \rho\bigr) = \Pr(Z_1 = 1) \bigl(\alpha + (1 - \alpha)q^{t-1} \bigr)\;.
\end{align*}
Similarly, for $t > s > 1$, 
\begin{align*}
  \E Z_s Z_t &= \Bigl(q^{s-1} \Pr(Z_1 = 1) + \sum_{i=1}^{s-1} A(i+1, s-1) \rho \Bigr) \bigl(q^{t-s} + \sum_{i=s}^{t-1} A(i+1, t-1) \rho\bigr) \\
  &= \Bigl(q^{s-1} \Pr(Z_1 = 1) + (1 - q^{s-1})\alpha \Bigr) \bigl(q^{t-s} + (1 - q^{t-s}) \alpha\bigr)\\
  &= \Bigl(\alpha + (\Pr(Z_1 = 1) - \alpha) q^{s-1} \Bigr) \bigl(\alpha + (1 - \alpha)q^{t-s} \bigr)\;.
\end{align*}

When $\rho = \gamma = 0$, the design satisfies $Z_T = Z_{T-1} = \dots = Z_1$, and therefore
\[
\E Z_s = \Pr(Z_1 = 1)\;, \quad \E Z_s Z_t = \Pr(Z_1 = 1)\;.
\]

Lastly, since $\Pr(Z_1 = 1) = \alpha$ by Definition \ref{def:soft}, one can easily verify that
\[
  \E Z_s Z_t =  \alpha (\alpha + (1 - \alpha)q^{t-s} )\;.
\]
\end{proof}

Then, the matrix $\E \Sigma$ in Proposition \ref{prop:sigma_soft} can be simplified as functions of $(\rho, \gamma)$. The following result establishes a closed-form expression for $\E \Sigma$ under random-switch designs, thereby leading to efficient evaluation of the design objective $\cL_{asy}(\rho, \gamma; w)$. In the statement, we use $\alpha = \rho / (\gamma + \rho)$ and $q = 1 - \gamma - \rho$. 

\begin{lemma}\label{lem:homog_M}
Suppose $T \ge 2K-1$. Then, for any $1\le i, j\le K$, the random-switch design with parameters $(\rho, \gamma)$ satisfies
\begin{equation}\label{eq:soft_obj}
\E \Sigma_{ij} = \alpha (1-\alpha) \Bigl(q^d - \frac{R_{ij}(q)}{T_\text{eff}^2}\Bigr)\;,
\end{equation}
where $d = |i - j|$ and
\[
R_{ij}(q) = \frac{(T_\text{eff} - d)(1-q^2) - 2q^{d+1} + q^{T_\text{eff}-d+1} + q^{T_\text{eff} + d + 1} }{(1 - q)^2}\;.
\]
In the extreme case of $\rho = \gamma = 0$, so that $q = 1$, Equation \eqref{eq:soft_obj} is defined in the left limit sense with $q \to 1^-$.
\end{lemma}

Lemma~\ref{lem:homog_M} characterizes the asymptotic objective $\cL_{asy}$
through the expected Gram matrix $\E\Sigma$, thereby reducing the design
optimization problem to an optimization over the Markov parameters
$(\rho,\gamma)$. The formula also reveals the structure of $\E\Sigma$. In
particular, the design performance separates into two components:
\[
\underbrace{\alpha(1-\alpha)}_{\text{marginal treatment variation}}
\underbrace{\Bigl(
q^d-\frac{R_{ij}(q)}{T_\text{eff}^2}
\Bigr)}_{\text{persistence}}\;.
\]
By definition, $\alpha$ is the marginal treatment probability under the homogeneous Markov chain. The parameter $q$ measures the persistence of treatment assignments: fewer switches, corresponding to smaller values of $\rho$ and $\gamma$, lead to larger values of $q$. Thus, for the asymptotic objective, the roles of the marginal treatment probability $\alpha$ and the persistence parameter $q$ can be analyzed separately. Moreover, Equation~\eqref{eq:soft_obj} decomposes the design information into a leading term and a finite-sample correction. When $|q|<1$ and $T_\text{eff}\to\infty$, the correction term is of smaller order, and the leading term $\alpha(1-\alpha)q^d$ dominates.

\begin{proof}[(Proof of Lemma \ref{lem:homog_M})]
Based on Lemma \ref{lem:homog_moments}, for any $1 \le s, t \le T$, we have
\[
\E Z_s Z_t = \alpha^2  + \alpha (1 - \alpha) q^{|s-t|}\;.
\]
Without loss of generality, we derive the expression of $\E \Sigma_{ij}$ with $i < j$. The $i > j$ case follows exactly the same argument, and the $i = j$ case can be verified similarly. 

For $i < j$, we apply the identity from Lemma \ref{lem:homog_moments} to obtain
\begin{align*}
\E \Sigma_{ij} &= \frac{1}{T_\text{eff}}\sum_{t=K}^{T} \E Z_{t-i+1} Z_{t-j+1} -
\frac{1}{T_\text{eff}^2} \sum_{s=K}^{T}\sum_{t=K}^{T} \E Z_{s-i+1} Z_{t-j+1} 
= \alpha (\alpha + (1-\alpha) q^{j-i} ) - \frac{1}{T_\text{eff}^2}
\sum_{s=K}^{T} \sum_{t=K}^{T} \alpha \bigl(\alpha + (1 - \alpha) q^{|t-s-j+i|}\bigr) \\
&= \alpha (\alpha + (1-\alpha) q^{j-i} ) - \alpha^2 - \frac{\alpha(1 - \alpha)}{T_\text{eff}^2}
\sum_{s=K}^{T} \sum_{t=K}^{T} q^{|t-s-j+i|} 
= \alpha (1-\alpha) q^{j-i} - \frac{\alpha(1 - \alpha)}{T_\text{eff}^2}
\sum_{s=K}^{T} \sum_{t=K}^{T} q^{|t-s-j+i|}\;.
\end{align*}
To simplify the last term, we apply the change of variable $r = t - s$ and $d = j - i$. By definition, $r$ ranges from $-(T_\text{eff}-1)$ to $T_\text{eff} - 1$. In particular, for each possible value of $r$, the summation contains $T_\text{eff} - |r|$ terms of $q^{|r - d|}$. This leads to 
\[
\sum_{s=K}^{T} \sum_{t=K}^{T} q^{|t-s-j+i|} = \sum_{r = -(T_\text{eff}-1)}^{T_\text{eff}-1} (T_\text{eff} - |r|) q^{|r - d|}\;. 
\]
Since $T\ge 2K - 1$, we have $T_\text{eff} - 1 = T - K \ge K-1 \ge d$. Therefore, we split the summation at $r = 0$ and $r = d$ to obtain
\[
\sum_{s=K}^{T} \sum_{t=K}^{T} q^{|t-s-j+i|} = \sum_{r = -(T_\text{eff}-1)}^{-1} (T_\text{eff} + r) q^{d - r} + \sum_{r = 0}^{d} (T_\text{eff} - r) q^{d - r} + \sum_{r = d+1}^{T_\text{eff}-1} (T_\text{eff} - r) q^{r - d}\;.
\]
For each term, we apply the formula for arithmetico-geometric sums and geometric sums to obtain
\begin{align*}
  \sum_{r = -(T_\text{eff}-1)}^{-1} (T_\text{eff} + r) q^{d - r} &= \frac{q^{d+1} ((T_\text{eff}-1)-T_\text{eff}q + q^{T_\text{eff}})}{(1 - q)^2}\;, \\
  \sum_{r = 0}^{d} (T_\text{eff} - r) q^{d - r} &= \frac{(T_\text{eff}-d) - (T_\text{eff}-d-1) q -(T_\text{eff}+1)q^{d+1} + T_\text{eff} q^{d+2} }{(1 - q)^2}\;, \\
  \sum_{r = d+1}^{T_\text{eff}-1} (T_\text{eff} - r) q^{r - d} &= \frac{(T_\text{eff}-d-1) q - (T_\text{eff}-d) q^2 + q^{T_\text{eff}-d+1}}{(1 - q)^2} \;. \\
\end{align*}
Summing up the three terms above, we have
\[
\sum_{s=K}^{T} \sum_{t=K}^{T} q^{|t-s-j+i|} = \frac{(T_\text{eff} - d)(1-q^2) - 2q^{d+1} + q^{T_\text{eff}-d+1} + q^{T_\text{eff}+d+1} }{(1 - q)^2}
\]
and for any $i < j$, 
\[
\E \Sigma_{ij} = \alpha (1-\alpha) q^{j-i} - \frac{\alpha(1 - \alpha)}{T_\text{eff}^2} \frac{(T_\text{eff} - d)(1-q^2) - 2q^{d+1} + q^{T_\text{eff}-d+1} + q^{T_\text{eff}+d+1} }{(1 - q)^2}\;, \quad d = j - i\;.
\]

In the $\rho = \gamma = 0$ case and $\Pr(Z_1 = 1) = \alpha$, we have $Z_1 = \dots = Z_T$ and 
\[
\E \Sigma_{ij} = \alpha - \frac{1}{T_\text{eff}^2} T_\text{eff}^2 \alpha = 0\;.
\]
One can verify by L'H\^{o}pital's rule that this value corresponds to the left limit of the expression derived above, since
\[
\lim_{q\to 1^-} \E \Sigma_{ij} = 0\;.
\]
\end{proof}

\begin{proof}[(Proof of Theorem \ref{thm:target_random})]
Define the feasible region in \eqref{eq:target_soft} to be
\[
\Theta_\delta=[\delta,1-\delta]^2\;.
\]
For any $(\rho, \gamma) \in \Theta_{\delta}$, one can verify that $\alpha = \rho / (\rho + \gamma)$ is uniformly bounded away from zero and one, and $q = 1 - \rho - \gamma$ is uniformly bounded away from $\pm 1$. 
Based on Lemma \ref{lem:homog_M}, for any $1\le i,j\le K$,
\[
\E \Sigma_{ij} = \alpha (1-\alpha) q^{|i - j|} - \alpha (1 - \alpha) \frac{R_{ij}(q)}{T_\text{eff}^2}\;.
\]
Therefore,
\[
\E \Sigma = \alpha(1-\alpha)\Sigma_K(q)+R_T(\rho,\gamma)\;, \qquad \Sigma_{K,ij}(q)=q^{|i-j|}\;.
\]
Here, $\Sigma_K(q)$ is the $K\times K$ Kac--Murdock--Szego (KMS) matrix \parencite{kac1953eigen}. By the uniform boundedness of $\alpha$ and $q$, the remainder term $R_T(\rho,\gamma)$ satisfies
\[
\sup_{(\rho,\gamma)\in \Theta_{\delta}}\|R_T(\rho,\gamma)\|\to 0\;.
\]

By the spectral theory of KMS matrices \parencite{dow2002explicit}, $\alpha(1-\alpha)\Sigma_K(q)$ is uniformly positive definite on the set $\Theta_{\delta}$. 
For all sufficiently large $T$, the uniform convergence of $R_T(\rho, \gamma)$ implies that $\E \Sigma$ is positive definite on $\Theta_{\delta}$. Hence, 
\[
\sup_{(\rho,\gamma)\in\Theta_{\delta}}
\Bigl\|
(\E \Sigma)^{-1} - \frac{1}{\alpha(1-\alpha)}\Sigma_K(q)^{-1}
\Bigr\| \to 0\;.
\]
Hence, for the targeted objective $l_T(\rho,\gamma; w) \coloneqq w^\top (\E\Sigma)^{-1}w$ and the limiting objective $l(\rho,\gamma;w) = \frac{1}{\alpha(1-\alpha)}w^\top\Sigma_K(q)^{-1}w$, we have the uniform convergence
\[
\sup_{(\rho,\gamma)\in\Theta_{\delta}}
|l_T(\rho, \gamma; w) - l(\rho, \gamma; w) | \to 0\;.
\]

The limiting objective separates into an $\alpha$-part and a $q$-part. For $\alpha$, the factor $1/\{\alpha(1-\alpha)\}$ is uniquely minimized at $\alpha^\star=1/2$. Thus the limiting optimizer must satisfy $\rho^\star=\gamma^\star$. For $q$, since $\Sigma_K(q)$ is the Kac-Murdock-Szego matrix, for $|q|<1$,
\[
\Sigma_K(q)^{-1} = \frac{1}{1-q^2}
\begin{pmatrix}
  1 & -q &  & & \\
  -q & 1+q^2 & -q & & \\
     & \ddots & \ddots & \ddots & \\
  && -q & 1+q^2 & -q \\
  &&& -q & 1
\end{pmatrix}\;.
\]
Therefore,
\[
w^\top \Sigma_K(q)^{-1} w = \frac{
\|w\|^2 + q^2\sum_{k=2}^{K-1}w_k^2 - 2q\sum_{k=1}^{K-1}w_kw_{k+1}
}{1-q^2}\;.
\]
Let
\[
B=\sum_{k=1}^{K-1}w_kw_{k+1}\;,
\qquad
D=w_1^2+w_K^2+2\sum_{k=2}^{K-1}w_k^2\;.
\]
The first-order condition with respect to $q$ is
\[
Bq^2-Dq+B=0 \;.
\]

We now show that the quadratic equation has a unique solution in $(-1,1)$ under the regularity assumption on $w$. First,
\[
D+2B=\sum_{k=1}^{K-1}(w_k+w_{k+1})^2\ge0\;,
\qquad
D-2B=\sum_{k=1}^{K-1}(w_k-w_{k+1})^2\ge0\;.
\]
Under the regularity conditions on $w$, both quantities are strictly positive, so
\[
D>2|B|\;.
\]
If $B=0$, the first-order condition reduces to $-Dq=0$, and hence
$q^\star=0$. If $B\neq0$, the two roots satisfy
\[
q_1q_2=1
\]
by Vieta's formula. Since $D>2|B|$, the roots are real and exactly one of them
lies in $(-1,1)$. Denote this unique root by $q^\star$. Because
$w^\top\Sigma_K(q)^{-1}w\to\infty$ as $q\to\pm1$ in the interior case,
this unique stationary point is the unique minimizer of the $q$-part.

Thus the limiting objective $l(\rho,\gamma;w)$ has the unique minimizer
\[
\alpha^\star=\frac12\;,
\qquad
q^\star\in(-1,1)\;,
\]
where $q^\star$ is the unique solution in $(-1,1)$ to $Bq^2-Dq+B=0$. Equivalently, $\rho^\star=\gamma^\star$, and substituting
$q^\star=1-2\rho^\star$ into the first-order condition gives
\[
4(\rho^\star)^2\sum_{k=1}^{K-1}w_kw_{k+1} + (2\rho^\star-1)\sum_{k=1}^{K-1}(w_k-w_{k+1})^2 = 0\;.
\]

It remains to justify the convergence of the finite-$T$ minimizers. Since $(\rho^\star, \gamma^\star) \in (\delta, 1 - \delta)^2$ and 
\[
\sup_{(\rho,\gamma)\in\Theta_{\delta}}
|l_T(\rho, \gamma; w) - l(\rho, \gamma; w) | \to 0\;,
\]
we apply the argmin theorem \parencite[Section 3.2.1]{wellner1996weak} to obtain $(\rho_T^\star, \gamma_T^\star) \to (\rho^\star, \gamma^\star)$.
\end{proof}

\begin{proof}[(Proof of Theorem \ref{thm:robust_random})]
Define the feasible region in \eqref{eq:robust_soft} to be
\[
\Theta_\delta=[\delta,1-\delta]^2\;,
\]
so that the robust design optimization \eqref{eq:robust_soft} reduces to 
\[
\min_{(\rho, \gamma) \in \Theta_\delta} \max_{\|w\|\le 1}  w^\top (\E \Sigma)^{-1} w \;.
\]

By Lemma \ref{lem:homog_M}, for $d=|i-j|$,
\[
\E\Sigma_{ij}
=
\alpha(1-\alpha)
\left(
q^d-\frac{R_{ij,T}(q)}{T_{\mathrm{eff}}^2}
\right)\;,\quad R_{ij,T}(q)
=
\frac{(T_{\mathrm{eff}}-d)(1-q^2)-2q^{d+1}
+q^{T_{\mathrm{eff}}-d+1}+q^{T_{\mathrm{eff}}+d+1}}
{(1-q)^2}\;.
\]
Since $q$ is uniformly bounded away from $1$ on $\Theta_\delta$, there exists a constant $C<\infty$ such that, uniformly over
$(\rho,\gamma)\in\Theta_\delta$ and $1\le i,j\le K$,
\[
|R_{ij,T}(q)|\le C T_{\mathrm{eff}}\;.
\]
Therefore,
\[
\sup_{(\rho,\gamma)\in\Theta_\delta}
\left\| \E\Sigma - \alpha(1-\alpha)\Sigma_K(q) \right\| \to 0\;,
\]
where $\Sigma_K(q)$ is the $K\times K$ Kac--Murdock--Szego (KMS) matrix \parencite{kac1953eigen} with $\Sigma_{K,ij}(q)=q^{|i-j|}$. By the spectral theory of KMS matrices \parencite{dow2002explicit}, $\Sigma_K(q)$ is positive definite uniformly on $\Theta_\delta$, which leads to
\[
\sup_{(\rho,\gamma)\in\Theta_\delta}
\left\| (\E\Sigma)^{-1} - \frac{1}{\alpha(1-\alpha)}\Sigma_K(q)^{-1} \right\| \to 0\;.
\]
Let $f_T(\rho, \gamma) = \max_{\|w\|\le 1} w^\top (\E\Sigma)^{-1} w$ and $f(\rho, \gamma) = \max_{\|w\|\le 1} w^\top\Sigma_K(q)^{-1} w/\alpha(1-\alpha)$. By Weyl's inequality \parencite{horn2012matrix} and the convergence result above, we have
\[
\sup_{(\rho,\gamma)\in\Theta_\delta}
\left|
f_T(\rho,\gamma)
-
f(\rho,\gamma)
\right|
\to 0\;.
\]

It remains to characterize the unique minimizer of $f$. First, by definition of the spectral norm, we have
\[
f(\rho, \gamma) = \frac{1}{\alpha(1 - \alpha)} \|\Sigma_K(q)^{-1}\| = \frac{1}{\alpha (1 - \alpha) \lambda_{\min}(\Sigma_K(q))}\;.
\]
For the term involving $\alpha$, observe that 
\[
\frac{1}{\alpha(1 - \alpha)} \ge 4\;,
\]
with equality if and only if $\alpha=1/2$. Second,
\[
\frac{1}{\lambda_{\min}(\Sigma_K(q))}
\ge
\frac{K}{ \operatorname{tr}(\Sigma_K(q))} = 1\;.
\]
For $K\ge2$, equality holds if and only if $\Sigma_K(q)=I_K$, which is equivalent to $q=0$. Hence the unique minimizer of $f$ on $\Theta_\delta$ is given by
\[
\alpha^\star=\frac12,
\qquad
q^\star=0.
\]
Using
\[
\rho=\alpha(1-q),
\qquad
\gamma=(1-\alpha)(1-q),
\]
we obtain
\[
\rho^\star=\gamma^\star=\frac12.
\]

Now let $(\rho_T^\star,\gamma_T^\star)$ be any minimizer of $f_T$ over $\Theta_\delta$. Since $f_T\to f$ uniformly and $f$ has the unique minimizer $(1/2,1/2)$, the argmin theorem \parencite[Section~3.2.1]{wellner1996weak} implies
\[
(\rho_T^\star,\gamma_T^\star)\to(1/2,1/2)\;.
\]
\end{proof}

\section{Proofs of Section \ref{sec:exact}}\label{sec:proof_exact}
To prove Proposition \ref{prop:sigma_hard}, we introduce the stronger result below and give the proof.
\begin{proposition}\label{prop:sigma_hard_supp}
Under the cycle-switch design with parameter $l$ and $T$, we have
\[
\Sigma_{ij} = f(i-j) - \frac{1}{4} + M_{ij}\;,
\]
where $M_{ij}$ converges to zero as $T$ goes to infinity uniformly over all $i, j$. Moreover, $M_{ij} = 0$ when $T_\text{eff}$ is divisible by $2l$. 
\end{proposition}

\begin{proof}[(Proof of Proposition \ref{prop:sigma_hard_supp})]
First, for any $i, j$, we have
\[
\Sigma_{ij} = \frac{1}{T_\text{eff}} \sum_{t=K}^{T} Z_{t-i+1} Z_{t-j+1} - \Bigl(\frac{1}{T_\text{eff}} \sum_{s=K}^{T}Z_{s-i+1} \Bigr) \Bigl(\frac{1}{T_\text{eff}} \sum_{t=K}^{T} Z_{t-j+1}\Bigr)\;.
\]
Therefore, it suffices to figure out the summations in the expression above. 
Let $n = l \lfloor \frac{T_\text{eff}}{2l} \rfloor$ denote the total number of treatments in full cycles of the cycle-switch design. We have
\[
\sum_{t=K}^{T} Z_{t-j+1} = n + R_j\;, \quad R_j = \sum_{t=K}^{K+r-1} Z_{t-j+1}\;, \quad r = T_\text{eff} \bmod 2l\;.
\]
Similarly, for the first term in the formula of $\Sigma_{ij}$, we have
\[
\sum_{t=K}^{T} Z_{t-i+1} Z_{t-j+1} = n \frac{|(i-j)\bmod 2l - l|}{l} + R_{ij}\;, \quad R_{ij}  = \sum_{t=K}^{K+r-1} Z_{t-j+1} Z_{t-i+1} \;.
\]
Summing up the derivations above, we have
\begin{align*}
  \Sigma_{ij} &= \frac{2n}{T_\text{eff}} f(i-j) - \Bigl(\frac{n}{T_\text{eff}}\Bigr)^2 + \frac{R_{ij}}{T_\text{eff}} - \frac{R_i+R_j}{T_\text{eff}} \frac{n}{T_\text{eff}} - \frac{R_i R_j}{T_\text{eff}^2}\;.
\end{align*}
When $T_\text{eff}$ is divisible by $2l$, we obtain $n = T_\text{eff} / 2$, $R_i = R_{ij} = 0$, and hence
\[
\Sigma_{ij} = f(i - j) - \frac{1}{4} \;.
\]
More generally, given $T_\text{eff}\to \infty$ and fixed $K, l$, one can verify that 
\[
\max_{i, j} \Bigl|\Sigma_{ij} - f(i - j) + \frac{1}{4}\Bigr| \to 0\;.
\]
\end{proof}

Let $\widetilde{\Sigma}$ be the limiting covariance matrix defined in Equation \eqref{eq:tilde_sigma}. The following result characterizes the inverse of $\widetilde{\Sigma}$, which is used in the analysis of design optimization below. 
\begin{lemma}\label{lem:inverse}
For $l \ge K$, it holds that
\[
\widetilde{\Sigma}^{-1} = l L_K + \frac{l}{l - K + 1} vv^\top\;,
\]
where $v = e_1 + e_K = (1, 0, \dots, 0, 1) \in \R^K$, and $L_K$ is the Laplacian defined by 
\[
L_K =
\begin{bmatrix}
1 & -1 & 0 & 0 & \cdots & 0 & 0\\
-1 & 2 & -1 & 0 & \cdots & 0 & 0\\
0 & -1 & 2 & -1 & \ddots & \vdots & \vdots\\
0 & 0 & -1 & 2 & \ddots & 0 & 0\\
\vdots & \vdots & \ddots & \ddots & \ddots & -1 & 0\\
0 & 0 & \cdots & 0 & -1 & 2 & -1\\
0 & 0 & \cdots & 0 & 0 & -1 & 1
\end{bmatrix}\;.
\]
\end{lemma}

\begin{proof}[(Proof of Lemma \ref{lem:inverse})]
By defining a matrix $D_K \in \R^{K\times K}$ via $(D_{K})_{ij} = |i - j|$, we have 
\[
\widetilde{\Sigma} = - \frac{1}{2l} D_K + \frac{1}{4} \bone_K \bone_K^\top\;.
\]
We first invoke the Sherman-Morrison-Woodbury identity to obtain
\[
\widetilde{\Sigma}^{-1} = \Bigl(- \frac{1}{2l} D_K + \frac{1}{4} \bone_K \bone_K^\top\Bigr)^{-1} = - 2l D_K^{-1} - \frac{1}{4 - 2l \bone_K^\top D_K^{-1} \bone_K} (2l)^2 D_K^{-1} \bone_K \bone_K^\top D_K^{-1}\;.
\]
Next, the distance matrix $D_K$ has a closed-form expression for its inverse. Specifically, for $K\ge 3$, we apply Theorem 8.9 of \cite{bapat2010graphs} to obtain
\[
D_K^{-1} = 
\begin{pmatrix}
- \frac{K-2}{2K - 2} & \frac{1}{2} & 0 & \dots & 0 & \frac{1}{2K-2} \\
\frac{1}{2} & -1 & \frac{1}{2} & \ddots &  & 0 \\  
0 & \frac{1}{2} & -1 & \frac{1}{2} & \ddots & 0 \\
\vdots & \ddots & \ddots & \ddots & \ddots & \vdots \\
0 & & \ddots & \frac{1}{2} & -1 & \frac{1}{2} \\
\frac{1}{2K-2} & 0 & \dots & 0 & \frac{1}{2} & - \frac{K-2}{2K-2} \\
\end{pmatrix}\;.
\]
Therefore, we have
\[
\bone_K^\top D_K^{-1} \bone_K = \frac{2}{K-1}\;, \quad D_K^{-1} \bone_K = \frac{1}{(K-1)} v\;.
\]
and the inverse of $\widetilde{\Sigma}$ simplifies to
\[
l L_K + \frac{l}{l - K + 1} vv^\top\;.
\]
The $K = 2$ case can be analyzed in a similar way using the fact that $D_2 = \begin{bmatrix}
  0 & 1 \\
  1 & 0
\end{bmatrix} = D_2^{-1}$.
\end{proof}

\subsection{Analysis of Targeted Design Optimization}\label{sec:proof_hard_target}
We introduce a few key lemmas that are used in the main proofs. The lemma below shows that $\widetilde{\Sigma} \in \R^{K\times K}$ has rank $\min\{l, K\}$. 

\begin{lemma}\label{lem:rank}
The covariance matrix $\widetilde{\Sigma}$ satisfies $\mathrm{rank}(\widetilde{\Sigma}) = \min\{l, K\}$.
\end{lemma}

Lemma \ref{lem:rank} entails two different settings. When $l < K$, $\widetilde{\Sigma}$ is rank-deficient and we can only detect some directions of the signal. This implies that the worst-case variance $\max_w w^\top \widetilde{\Sigma}^{-1} w$ is $\infty$. When $l \ge K$, the $\widetilde{\Sigma}$ is full-rank, and $l$ could potentially be the optimizer to the robust design optimization \eqref{eq:robust_hard}.

\begin{proof}[(Proof of Lemma \ref{lem:rank})]
In the following, we prove that $\mathrm{rank}(\widetilde{\Sigma}) = l$ when $l < K$. The $l \ge K$ setting can be proved via a similar argument. For $\widetilde{\Sigma}_{ij} = f(i - j) - \frac14$, we show that for all integers $d$, 
\[
f(d + l) + f(d) = \frac12\;.
\]
Let $r = d \bmod 2l \in \{0,\dots,2l-1\}$. If $0 \le r < l$, then
\[
f(d) = \frac{l-r}{2l}\;,
\qquad
f(d+l) = \frac{r}{2l}\;,
\]
leading to $f(d) + f(d+l) = 1/2$. Similarly, if $l \le r < 2l$, then
\[
f(d) = \frac{r-l}{2l}\;, \qquad
f(d+l) = \frac{2l-r}{2l}\;.
\]
Hence in all cases, we have $f(d) + f(d+l) = 1/2$. Since $K > l$, we set $i$ to be a positive integer such that $1 \le i \le K-l$. Then, for any $j = 1, \dots, K$, the identity above implies the following sign-flip identity: 
\[
\widetilde{\Sigma}_{i + l, j} = f(i + l - j) - \frac{1}{4} = - f(i - j) + \frac{1}{4} = - \widetilde{\Sigma}_{i, j}\;.
\]
Thus the $(i+l)$-th row is exactly the negative of the $i$-th row. It follows that all rows of
$\widetilde{\Sigma}$ are generated by the first $l$ rows, and hence we derive the upper bound on the rank of $\widetilde{\Sigma}$, i.e., $\mathrm{rank}(\widetilde{\Sigma}) \le l$.

For the reverse inequality, consider the leading $l\times l$
principal submatrix $H=(\widetilde{\Sigma}_{ij})_{1\le i,j\le l}$. Since
$|i-j|\le l-1$ on this range, we have
\[
H_{ij}=\frac14-\frac{|i-j|}{2l}.
\]
Now define the $l\times l$ matrix
\[
B=
\begin{bmatrix}
2 & -1 & 0 & \cdots & 0 & 1\\
-1 & 2 & -1 & \ddots & & 0\\
0 & -1 & 2 & \ddots & \ddots & \vdots\\
\vdots & \ddots & \ddots & \ddots & -1 & 0\\
0 & & \ddots & -1 & 2 & -1\\
1 & 0 & \cdots & 0 & -1 & 2
\end{bmatrix}.
\]
A direct calculation shows that $H(lB)=I_l$, and hence $H^{-1}=lB$. Therefore $H$ is invertible, so $\mathrm{rank}(H)=l$. Since $H$ is a principal
submatrix of $\widetilde{\Sigma}$, this implies $\mathrm{rank}(\widetilde{\Sigma})\ge l$. Combining the two bounds yields
\[
\mathrm{rank}(\widetilde{\Sigma}) = l\;.
\]
\end{proof}

\begin{proof}[(Proof of Theorem \ref{thm:identification})]
Define the matrix $A\in\mathbb R^{K\times l}$ by
\[
A_{u+rl,u}=(-1)^r
\]
for $u=1,\dots,l$ and $r=0,\dots,\lfloor (K-u)/l\rfloor$, and set
$A_{ij}=0$ otherwise. Equivalently,
\[
A=[b_1,\dots,b_l]\;,
\]
where $b_u$ is defined in Theorem \ref{thm:identification}. 

Let
\[
H=(\widetilde{\Sigma}_{ij})_{i,j\in[l]}\;,
\qquad [l]=\{1,\dots,l\}\;.
\]
By the sign-flip identity established in the proof of Lemma \ref{lem:rank}, one can verify that 
\[
\widetilde{\Sigma}=AHA^\top\;.
\]

The columns of $A$ have disjoint supports and are therefore linearly
independent, i.e., $\mathrm{rank}(A) = l$.
Moreover, $H$ is invertible by the proof of Lemma \ref{lem:rank}. It follows
that
\[
\mathrm{rank}(\widetilde{\Sigma}) = \mathrm{rank}(AHA^\top) = l\;.
\]

Since $\widetilde{\Sigma}=AHA^\top$, we have $\mathrm{Im}(\widetilde{\Sigma})\subseteq \mathrm{Im}(A)$. Since both image spaces have dimension $l$, we have
\[
\mathrm{Im}(\widetilde{\Sigma}) = \mathrm{Im}(A) = \mathrm{span}\{b_u:u=1,\dots,l\}\;.
\]

Finally, an estimand vector $w$ is estimable if and only if
\[
w\in\operatorname{Im}(\widetilde{\Sigma}).
\]
Therefore, $w$ is estimable if and only if
\[
w\in\operatorname{span}\{b_u:u=1,\dots,l\}\;.
\]
\end{proof}

\begin{proof}[(Proof of Theorem \ref{thm:exact_sol2})]
Based on Proposition \ref{prop:sigma_hard}, for any $l$, we have
\[
\widetilde{\Sigma}_{ij} = f(i - j) - \frac{1}{4}\;,
\] 
with the function $f$ defined in Equation \eqref{eq:fq}. Moreover, since $l \ge K$, we obtain
\[
f(i - j) = \frac{l - |i - j|}{2l} = \frac{1}{2} - \frac{|i - j|}{2l}\;.
\]

By Lemma \ref{lem:inverse}, we have
\[
\widetilde{\Sigma}^{-1} = l L_K + \frac{l}{l - K + 1} vv^\top\;,
\]
where $v = e_1 + e_K$, and $L_K$ is the Laplacian defined by 
\[
L_K =
\begin{bmatrix}
1 & -1 & 0 & 0 & \cdots & 0 & 0\\
-1 & 2 & -1 & 0 & \cdots & 0 & 0\\
0 & -1 & 2 & -1 & \ddots & \vdots & \vdots\\
0 & 0 & -1 & 2 & \ddots & 0 & 0\\
\vdots & \vdots & \ddots & \ddots & \ddots & -1 & 0\\
0 & 0 & \cdots & 0 & -1 & 2 & -1\\
0 & 0 & \cdots & 0 & 0 & -1 & 1
\end{bmatrix}\;.
\]
Based on the matrix inverse, the objective reduces to
\[
w^\top \widetilde{\Sigma}^{-1} w = l w^\top L_K w + \frac{l}{l - K + 1} (w^\top v)^2 \eqqcolon l A(w) + \frac{l}{l - K + 1} B(w)\;,
\]
where 
\begin{align*}
  A(w) &= w^\top L_K w = \sum_{k=1}^{K-1} (w_k - w_{k+1})^2\;, \\
  B(w) &= (w_1 + w_K)^2\;.
\end{align*}
This completes the proof.
\end{proof}

Lastly, we prove Corollaries \ref{cor:cumulative} and \ref{cor:lag}, which characterize the optimal designs under specific classes of estimands.

\begin{proof}[(Proof of Corollary \ref{cor:cumulative})]
The $l \ge K$ setting follows from a direct application of Theorem \ref{thm:exact_sol2}, and therefore it suffices to prove that $w$ is not estimable when $l < K$. First consider the case $l<K$. By the sign-flip identity for $\widetilde{\Sigma}$ (Lemma \ref{lem:rank}), we have
\[
\widetilde{\Sigma}_{i+l,j}=-\widetilde{\Sigma}_{ij}\;,
\qquad i=1,\dots,K-l\;.
\]
For any vector $a\in\mathbb R^K$,
\[
(\widetilde{\Sigma}a)_{i+l} = -(\widetilde{\Sigma}a)_i\;,
\qquad i=1,\dots,K-l\;.
\]
Therefore every vector in $\mathrm{Im}(\widetilde{\Sigma})$ satisfies the alternating constraint
\[
x_{i+l}=-x_i\;, \qquad i=1,\dots,K-l\;.
\]
In particular, since $l<K$, the coordinate $1+l$ exists. For the cumulative effect vector $w$, we have $w_1=1$, while
\[
w_{1+l}\in\{0,1\}\;.
\]
Thus
\[
w_{1+l}\neq -w_1\;.
\]
Hence $w\notin\mathrm{Im}(\widetilde{\Sigma})$. Therefore the cumulative effect is not estimable when $l<K$, and under our convention the corresponding variance is infinite:
\[
w^\top\widetilde{\Sigma}^{-1}w=\infty\;.
\]
\end{proof}

\begin{proof}[(Proof of Corollary \ref{cor:lag})]
First consider the case $l<K$. By the sign-flip identity for the cycle-switch limiting covariance matrix (Lemma \ref{lem:rank}),
\[
\widetilde{\Sigma}_{i+l,j}=-\widetilde{\Sigma}_{ij}\;,
\qquad i=1,\dots,K-l\;.
\]
Since $\widetilde{\Sigma}$ is symmetric, every vector
$x\in\operatorname{Im}(\widetilde{\Sigma})$ satisfies
\[
x_{i+l}=-x_i\;,
\qquad i=1,\dots,K-l\;.
\]
Next, we focus on $w = e_k$ and study its estimation variance. First, if $k\le K-l$, then the constraint with $i=k$
gives $(e_k)_{k+l}=-(e_k)_k$, which implies $0=-1$, a contradiction. Hence we need $k>K-l$, or equivalently,
\[
l\ge K-k+1\;.
\]
Second, if $k>l$, then the constraint with $i=k-l$ gives $(e_k)_k=-(e_k)_{k-l}$, which implies $1=0$, again a contradiction. Hence we need
\[
k\le l\;.
\]
Combining the two necessary conditions, $e_k$ is estimable only if
\[
l\ge \max\{k,K-k+1\}\;.
\]
If this condition fails, then $e_k\notin\operatorname{Im}(\widetilde{\Sigma})$, and under our convention the corresponding variance is infinite:
\[
e_k^\top\widetilde{\Sigma}^{-1}e_k=\infty\;.
\]

Now suppose
\[
l\ge \max\{k,K-k+1\}\;.
\]
If $l\ge K$, then Theorem \ref{thm:exact_sol2} gives
\[
e_k^\top\widetilde{\Sigma}^{-1}e_k
=
lA(e_k)+\frac{l}{l-K+1}B(e_k)\;.
\]
Since $k=2,\dots,K-1$, we have
\[
A(e_k)=2\;, \qquad B(e_k)=0\;.
\]
Therefore,
\[
e_k^\top\widetilde{\Sigma}^{-1}e_k = 2l\;.
\]

If $\max\{k,K-k+1\}\le l<K$, by Theorem \ref{thm:identification}, $e_k\in\operatorname{Im}(\widetilde{\Sigma})$, so the lag-specific effect is estimable. Moreover, the factorization used in
the proof of Theorem \ref{thm:identification} gives
\[
\widetilde{\Sigma}=AHA^\top\;,
\]
where $A=[b_1,\dots,b_l]$ and
\[
H=(\widetilde{\Sigma}_{ij})_{1\le i,j\le l}\;.
\]
Because $l\ge k$ and $l\ge K-k+1$, the coordinate $k$ is not paired with
any other coordinate through the sign-flip pattern. Hence the $k$-th column of
$A$ is exactly $e_k$. Therefore,
\[
e_k^\top\widetilde{\Sigma}^{-1}e_k=e_k^\top H^{-1}e_k\;.
\]
Here, the expression on the left hand side follows the convention of matrix inverse in Equation \eqref{eq:inv}. With an abuse of notation, $e_k$ on the left hand side above denotes the $k$-th basis vector in $\R^K$, whereas $e_k$ on the right hand side denotes the $k$-th basis vector in $\R^l$. From the proof of Lemma \ref{lem:rank}, we have
\[
H^{-1}=lB\;,
\]
where $B$ is defined in the proof of Lemma \ref{lem:rank} and has diagonal entries equal to $2$. Hence we have $e_k^\top H^{-1}e_k=2l$, which implies
\[
e_k^\top\widetilde{\Sigma}^{-1} e_k=2l\;.
\]
Combining the cases above proves the result.
\end{proof}

\subsection{Analysis of Robust Design Optimization}\label{sec:proof_hard_robust}

The following lemma characterizes a subset of eigenvalues of $\widetilde{\Sigma}^{-1}$, which will be used to derive the closed-form expression for the robust design objective. 

\begin{lemma}\label{lem:spectrum}
Suppose $l \ge K$. The matrix $\widetilde{\Sigma}^{-1}$ has $\lfloor \frac{K}{2} \rfloor$ eigenvalues defined by
\[
4l \sin^2\Bigl(\frac{(2r-1)\pi}{2K}\Bigr)\;, \quad r = 1, \dots, \lfloor \frac{K}{2} \rfloor\;.
\]
\end{lemma}

\begin{proof}[(Proof of Lemma \ref{lem:spectrum})]
We prove this lemma by way of construction. First, by Lemma \ref{lem:inverse}, we have
\[
\widetilde{\Sigma}^{-1} = l L_K + \frac{l}{l - K + 1} vv^\top\;,
\]
where $v = e_1 + e_K$, and $L_K$ is the Laplacian
\[
L_K =
\begin{bmatrix}
1 & -1 & 0 & 0 & \cdots & 0 & 0\\
-1 & 2 & -1 & 0 & \cdots & 0 & 0\\
0 & -1 & 2 & -1 & \ddots & \vdots & \vdots\\
0 & 0 & -1 & 2 & \ddots & 0 & 0\\
\vdots & \vdots & \ddots & \ddots & \ddots & -1 & 0\\
0 & 0 & \cdots & 0 & -1 & 2 & -1\\
0 & 0 & \cdots & 0 & 0 & -1 & 1
\end{bmatrix}\;.
\]

The matrix $L_K$ has explicit eigenvalues and eigenvectors as below. Define $q_1 = \bone_K / \sqrt{K}$ and 
\[
q_{kj} = \sqrt{\frac{2}{K}} \cos\Bigl( \frac{(k-1)(j-1/2)}{K} \pi \Bigr)\;.
\]
One can verify by trigonometric identities that $\{q_k\}_{k=1}^{K}$ forms an orthonormal basis of $\R^{K}$. Moreover, by some direct algebra, one can also verify that 
\[
L_K q_k = \mu_k q_k\;, \quad \mu_1 = 0\;, \quad \mu_k = 4 \sin^2\Bigl(\frac{(k-1)\pi}{2K}\Bigr)\;.
\]
By definition of $q_k$, we have $v^\top q_k = 0$ for all even $k$, i.e., $k = 2r$. Therefore, for any $r = 1, \dots, \lfloor \frac{K}{2} \rfloor$, the vector $q_{2r}$ remains an eigenvector for the matrix $\widetilde{\Sigma}^{-1}$ in the sense that
\[
\widetilde{\Sigma}^{-1} q_{2r} = l \mu_{2r} q_{2r} + 0 = l \mu_{2r} q_{2r}\;.
\]
Therefore, $\widetilde{\Sigma}^{-1}$ has $\lfloor \frac{K}{2} \rfloor$ eigenvalues defined by $l \mu_{2r}$, i.e., 
\[
4 l \sin^2\Bigl(\frac{(2r-1)\pi}{2K}\Bigr)\;, \quad r = 1, \dots, \lfloor \frac{K}{2} \rfloor\;.
\]
\end{proof}

Next, we prove Theorem \ref{thm:robust2} based on the lemmas above.
\begin{proof}[(Proof of Theorem \ref{thm:robust2})]
By Lemma \ref{lem:rank}, it suffices to focus on the case $l \ge K$. Based on Lemma \ref{lem:spectrum}, for even $K$, the matrix $\widetilde{\Sigma}^{-1}$ has eigenvalues
\[
4l \sin^2\Bigl(\frac{(2r-1)\pi}{2K}\Bigr)\;, \quad r = 1, \dots, \frac{K}{2} \;.
\]
By setting $r = K/2$, we obtain the eigenvalue 
\[
4l \sin^2\Bigl(\frac{(K-1)\pi}{2K}\Bigr) = 4l \cos^2\Bigl(\frac{\pi}{2K}\Bigr)\;.
\]
Therefore, we have 
\[
\max_{\|w\| \le 1} w^\top \widetilde{\Sigma}^{-1} w = \|\widetilde{\Sigma}^{-1}\| \ge 4l \cos^2\Bigl(\frac{\pi}{2K}\Bigr)\;.
\]

Next, we show that this eigenvalue is in fact an upper bound, i.e., $\|\widetilde{\Sigma}^{-1}\| \le 4l \cos^2\Bigl(\frac{\pi}{2K}\Bigr)$, and therefore it is exactly the spectral norm. From the proof of Lemma \ref{lem:spectrum}, we have
\[
\widetilde{\Sigma}^{-1} = l L_K + \frac{l}{l-K+1} vv^\top\;, \quad v = e_1 + e_K\;.
\]
Since $l \ge K$, we have $l / (l - K + 1) \le l$, and therefore
\[
\widetilde{\Sigma}^{-1} \preceq  l L_K + l vv^\top = l M_K\;, \quad M_K = L_K + vv^\top\;.
\]
By definition of $L_K$, we have 
\[
L_K + vv^\top = 
\begin{bmatrix}
2 & -1 & 0 & 0 & \cdots & 0 & 1\\
-1 & 2 & -1 & 0 & \cdots & 0 & 0\\
0 & -1 & 2 & -1 & \ddots & \vdots & \vdots\\
0 & 0 & -1 & 2 & \ddots & 0 & 0\\
\vdots & \vdots & \ddots & \ddots & \ddots & -1 & 0\\
0 & 0 & \cdots & 0 & -1 & 2 & -1\\
1 & 0 & \cdots & 0 & 0 & -1 & 2
\end{bmatrix}\;.
\]
Note that it is the anti-periodic discrete Laplacian, and its eigenvalues can be expressed as 
\[
4 \sin^2\Bigl(\frac{(2r-1)\pi}{2K}\Bigr)\;, \quad r = 1, \dots, K\;.
\]
Therefore, we have
\[
\|\widetilde{\Sigma}^{-1}\| \le l \|M_K\| = 4l \cos^2\Bigl(\frac{\pi}{2K}\Bigr)\;.
\]

Summarizing the derivations above, we have
\[
\max_{\|w\| \le 1} w^\top \widetilde{\Sigma}^{-1} w = \|\widetilde{\Sigma}^{-1}\| = 4l \cos^2\Bigl(\frac{\pi}{2K}\Bigr)\;.
\]
Since the worst-case variance grows linearly in $l$ for $l \ge K$, the optimizer is $l^\star = K$.
\end{proof}

\section{Connections to Design-based Inference}\label{sec:connection}
In this section, we provide a design-based interpretation of the OLS estimator as summarized in Remark \ref{rmk:est}. We first derive the design-based limit of the OLS estimator under random-switch designs as $T\to\infty$. This limit can be interpreted as a linear projection target determined by design-induced conditional-outcome contrasts. We then discuss two concrete settings under which this design-based limit recovers a causal interpretation studied by \textcite{lin2025unifying}.

\begin{assumption}\label{asmp:regular_rand}
The treatment assignments $\{Z_t\}_{t=1}^T$ follow a random-switch design with $\rho, \gamma \in (0, 1)$.
\end{assumption}
By Assumption \ref{asmp:regular_rand}, we focus on random-switch designs studied in Section \ref{sec:homog}, where $\rho_t = \rho$ and $\gamma_t = \gamma$ for $t = 1, \dots, T-1$. Furthermore, we assume that the potential outcomes satisfy fixed lag and non-anticipation, and that they are uniformly bounded.
\begin{assumption}\label{asmp:po}
The potential outcomes satisfy the following conditions:
\begin{enumerate}
  \item \emph{Non-anticipation and fixed lag:} for every $t\ge K$,
  \[
  Y_t(z_1,\dots,z_T)=Y_t(z_{t-K+1},\dots,z_t).
  \]
  \item \emph{Uniform boundedness:} there exists a constant $M>0$ such that
  \[
  |Y_t(z_{t-K+1},\dots,z_t)|\le M
  \]
  for every $t\ge K$ and every $(z_{t-K+1},\dots,z_t)\in\{0,1\}^K$.
\end{enumerate}
\end{assumption}

Let $Y_t^{\obs} = Y_t(Z_{t-K+1}, \dots, Z_t)$ be the observed outcome at time $t$, where $Z_{t-K+1}, \dots, Z_t$ are the treatment assignments from a random-switch design. For a fixed lag length \(K\), recall the lag-treatment vector
\[
X_t=(Z_t,Z_{t-1},\dots,Z_{t-K+1})^\top\in\R^K,
\qquad t=K,\dots,T\;,
\]
and $\bar X=\frac1{T_\text{eff}}\sum_{t=K}^T X_t$. The OLS estimate can be expressed as
\begin{equation*}
\widehat\beta =
\Bigl( \frac1{T_\text{eff}}\sum_{t=K}^T (X_t-\bar X)(X_t-\bar X)^\top \Bigr)^{-1} \Bigl( \frac1{T_\text{eff}}\sum_{t=K}^T (X_t-\bar X)Y_t^{\mathrm{obs}} \Bigr)\;.
\end{equation*}
The goal of this section is to derive the limit of $\widehat{\beta}$ under design-based inference.

We establish the following convergence results for the design-based moments, which follows from the ergodic theorem for Markov chains. For a sequence of random variables $\{X_t\}_{t=1}^\infty$, we write $X_t = o_p(1)$ to mean that $X_t$ converges to zero in probability as $t\to\infty$.
\begin{lemma}
\label{lem:LLN}
Suppose Assumptions \ref{asmp:regular_rand} and \ref{asmp:po} hold. We have
\begin{gather}
\frac1{T_\text{eff}}\sum_{t=K}^T
(X_t-\bar{X})(X_t-\bar{X})^\top = \alpha (1 - \alpha) \Sigma_K(q) + o_p(1)\;,\label{eq:LLN1}\\
\frac1{T_\text{eff}}\sum_{t=K}^T
 (X_t-\bar{X}) Y_t^{\obs} = \alpha (1 - \alpha) C_T + o_p(1) \label{eq:LLN2}\;.
\end{gather}
\end{lemma}
In Lemma \ref{lem:LLN}, we define the $K$-by-$K$ Toeplitz matrix $\Sigma_K(q)$ by
\[
\Sigma_{K, ij}(q) = q^{|i-j|}\;, \quad i, j = 1, \dots, K\;,
\]
where $q = 1 - \rho - \gamma$ as defined in Section \ref{sec:homog}. For each \(t\) and \(k=1,\dots,K\), define the design-induced conditional-outcome contrast
\[
C_{t,k} =
\E \Bigl(Y_t^{\mathrm{obs}}\mid Z_{t-k+1}=1\Bigr)
-
\E \Bigl(Y_t^{\mathrm{obs}}\mid Z_{t-k+1}=0\Bigr)\;.
\]
Then, the vector $C_T$ is defined as
\[
C_t=(C_{t,1},\dots,C_{t,K})^\top,
\qquad
C_T=
\frac1{T_\text{eff}}\sum_{t=K}^T C_t\;.
\]
Intuitively, $\Sigma_K(q)$ captures the limiting sample covariance of the lagged treatment vector, as analyzed in Section~\ref{sec:proof_homog}. The vector $C_T$ collects time-averaged contrasts between conditional outcomes. Throughout this section, expectations are taken only over the randomization design $\{Z_t\}_{t = K}^T$ with the potential outcomes held fixed.

Based on Lemma \ref{lem:LLN}, we derive the limit of the OLS estimator $\widehat{\beta}$ under the design-based framework.
\begin{theorem}[Design-based limit of OLS]
\label{thm:ols}
Suppose Assumptions \ref{asmp:regular_rand} and \ref{asmp:po} hold. Then
\[
\widehat\beta - \Sigma_K(q)^{-1}C_T=o_p(1) \;.
\]
\end{theorem}

\begin{proof}
By Lemma~\ref{lem:LLN},
\[
\frac1{T_\text{eff}}\sum_{t=K}^T
(X_t-\bar{X})(X_t-\bar{X})^\top = \alpha (1 - \alpha) \Sigma_K(q) + o_p(1)\;,\quad 
\frac1{T_\text{eff}}\sum_{t=K}^T (X_t-\bar{X}) Y_t^{\obs} = \alpha (1 - \alpha) C_T + o_p(1)\;.
\]
Since \(\alpha(1-\alpha)\Sigma_K(q)\) is nonsingular from the analysis of Section \ref{sec:proof_homog}, the continuous mapping theorem gives
\[
\Bigl( \frac1{T_\text{eff}}\sum_{t=K}^T
(X_t-\bar{X})(X_t-\bar{X})^\top \Bigr)^{-1} =
\frac{1}{\alpha(1-\alpha)} \Sigma_K(q)^{-1}+o_p(1)\;.
\]
Moreover, we have \(C_T=O(1)\) by Assumption~\ref{asmp:po}. Therefore, by standard stochastic order algebra,
\[
\widehat\beta = \Sigma_K(q)^{-1} C_T+o_p(1)\;.
\]
\end{proof}

Theorem~\ref{thm:ols} shows that, without imposing a linear outcome model, OLS converges to \(\Sigma_K(q)^{-1} C_T\). The vector \(C_T\) consists of conditional-outcome contrasts induced by the assignment design. These contrasts are generally different from marginal potential-outcome contrasts, because under a Markov design, conditioning on one treatment assignment changes the distribution of nearby treatment assignments.

Next, we discuss two concrete setups, where the design-based OLS limit in Theorem \ref{thm:ols} reduces to a lagged treatment-effect estimand $\beta^{\mathrm{LD}}$. Specifically, we consider the causal estimand studied in \textcite{lin2025unifying}. Under Assumption~\ref{asmp:po}, for lag \(k=1,\dots,K\), define
\[
\beta_{k}^{\mathrm{LD}}
=
\frac1{T_\text{eff}}\sum_{t=K}^T
\E \Bigl(
Y_t\bigl(Z_{t-K+1:t}^{(t-k+1)=1}\bigr)
-
Y_t\bigl(Z_{t-K+1:t}^{(t-k+1)=0}\bigr)
\Bigr)\;,\footnote{The superscript \(\mathrm{LD}\) refers to Lin and Ding.}
\]
where \(Z_{t-K+1:t}^{(s)=a}\) denotes the treatment path $(Z_{t-K+1}, \dots, Z_t)$ obtained by replacing \(Z_s\) with \(a\), while leaving all other assignments unchanged. The expectation is taken over the treatment distribution. Let
\[
\beta^{\mathrm{LD}}
=
(\beta_{1}^{\mathrm{LD}},\dots,\beta_{K}^{\mathrm{LD}})^\top.
\]
Intuitively, \(\beta^{\mathrm{LD}}\) summarizes the design-marginal lagged treatment effects under the given assignment design. 

\paragraph{Setup 1: i.i.d. Bernoulli randomization.}
The following result shows that, under i.i.d. Bernoulli randomization, our OLS projection target coincides with \(\beta^{\mathrm{LD}}\).

\begin{corollary}\label{cor:iid}
Suppose Assumptions \ref{asmp:regular_rand} and \ref{asmp:po} hold with $\rho + \gamma = 1$. Then
\[
\widehat\beta - \beta^{\mathrm{LD}} = o_p(1) \;.
\]
\end{corollary}
By definition of the random-switch design, $\rho+ \gamma = 1$ recovers the i.i.d. Bernoulli randomization with treatment probability $\rho$. That is, the OLS estimate consistently captures the design-based estimand $\beta^{\mathrm{LD}}$ under i.i.d. Bernoulli randomization. 

\paragraph{Setup 2: Linear additive potential outcomes.}
The first setup imposes independence on the design, whereas our second setup instead allows Markov dependence in the design but imposes additive linear structure on the potential outcomes.
\begin{assumption}
\label{asmp:additive-po}
For \(t=K,\dots,T\), the potential outcomes satisfy
\[
Y_t(z_{1}, \dots, z_t) = \mu_{t} + \sum_{j=1}^K \delta_{t,j}z_{t-j+1}\;,
\]
where \(\mu_{t}\) and \(\delta_{t,j}\) are fixed quantities satisfying
\[
\sup_t|\mu_t|<\infty\;,
\qquad
\sup_{t,j}|\delta_{t,j}|<\infty\;.
\]
\end{assumption}

Notably, Assumption \ref{asmp:additive-po} allows time-varying treatment effects as well as time fixed effects. In this sense, Assumption \ref{asmp:additive-po} can be viewed as a relaxation of the impulse-response model \eqref{eq:main_model}. In addition, Assumption~\ref{asmp:additive-po} implies Assumption~\ref{asmp:po}, since $K$ is fixed and the coefficients are uniformly bounded.

\begin{corollary}
\label{cor:additive}
Suppose Assumptions \ref{asmp:regular_rand} and \ref{asmp:additive-po} hold. Then,
\[
\widehat\beta-\beta^{\mathrm{LD}} = o_p(1)\;.
\]
Moreover, we have 
\[
\beta^{\mathrm{LD}}_k = \frac{1}{T_\text{eff}} \sum_{t=K}^T \delta_{t, k}\;.
\]
\end{corollary}

The corollary above shows that the OLS estimator retains a causal interpretation under additive linear potential outcomes, again recovering the causal estimand $\beta^{\mathrm{LD}}$ studied in \textcite{lin2025unifying}.

\subsection{Proofs}
\begin{proof}[(Proof of Lemma \ref{lem:LLN})]
\noindent \underline{Proof of Equation \eqref{eq:LLN1}. }
We first establish the limit of $\bar{X}$. Note that by definition, 
\[
\bar{X} = \frac{1}{T_\text{eff}} \sum_{t=K}^T X_t\;.
\]
By Definition \ref{def:soft}, a random-switch design $\{Z_t\}_{t=1}^T$ is initialized from the stationary distribution, and therefore $\E X_t = \alpha \bone_K$ for all $t$. To apply the law of large numbers, we note that $\rho, \gamma \in (0, 1)$ under Assumption \ref{asmp:regular_rand} and hence $\{Z_t\}_{t=1}^T$ is a two-state irreducible and aperiodic Markov chain \parencite{norris1998markov}. As a result, $\{X_t\}_{t=K}^T = \{(Z_{t}, \dots, Z_{t-K+1})\}_{t=K}^T$ is also an irreducible and aperiodic Markov chain. Therefore, by the law of large numbers for Markov chains (see, e.g., Theorem 17.0.1 of \textcite{meyn2012markov}), we have
\[
 \frac{1}{T_\text{eff}} \sum_{t=K}^T (X_t - \E X_t) = \frac{1}{T_\text{eff}} \sum_{t=K}^T X_t - \alpha \bone_K = o_p(1)\;.
\]
Here, $\bone_n$ is a length-$n$ all-one vector. 
Since $\bar X-\alpha\bone_K=o_p(1)$, replacing $\bar X$ by $\alpha\bone_K$ in Equation \eqref{eq:LLN1} only changes the left hand side by $o_p(1)$.
To prove Equation \eqref{eq:LLN1}, it then suffices to show the following two results: 
\begin{align}
\frac1{T_\text{eff}}\sum_{t=K}^T
(X_t-\alpha\bone_K )(X_t-\alpha \bone_K)^\top &= 
\E (X_T-\alpha\bone_K )(X_T-\alpha \bone_K)^\top + o_p(1)\;,\label{eq:LLN3} \\
\E (X_T-\alpha\bone_K )(X_T-\alpha \bone_K)^\top &= 
\alpha (1 - \alpha) \Sigma_K(q)\;.\label{eq:Sigma_eval}
\end{align}
Equation \eqref{eq:LLN3} also follows from the law of large numbers for Markov chains. Specifically, let 
\begin{align*}
& \frac1{T_\text{eff}} \sum_{t=K}^T
\left[
(X_t-\alpha\bone_K)(X_t-\alpha\bone_K)^\top -
\E \Bigl((X_t-\alpha\bone_K)(X_t-\alpha\bone_K)^\top\Bigr)
\right] = \frac{1}{T_\text{eff}} \sum_{t=K}^T g(X_t) - \E g(X_t)\;, \\
g(X_t) & = (X_t-\alpha\bone_K)(X_t-\alpha\bone_K)^\top\;.
\end{align*}
As explained above, $X_t$ is an irreducible and aperiodic Markov chain. In addition, the function $g$ is uniformly bounded on the support of $X_t$, i.e., $\{0, 1\}^K$. Then, we apply the law of large numbers (Theorem 17.0.1 of \textcite{meyn2012markov}) componentwise to obtain
\[
\frac{1}{T_\text{eff}} \sum_{t=K}^T \bigl(g(X_t) - \E g(X_t)\bigr) = o_p(1)\;.
\]
Since $X_t$ is stationary, $\E g(X_t)$ does not depend on $t$. Therefore, the law of large numbers implies Equation~\eqref{eq:LLN3}.
Additionally, applying Lemma \ref{lem:homog_moments} to the left hand side of Equation \eqref{eq:Sigma_eval} verifies Equation~\eqref{eq:Sigma_eval}.

\noindent \underline{Proof of Equation \eqref{eq:LLN2}.}
For Equation \eqref{eq:LLN2}, it again suffices to prove the following two results: 
\begin{align}
\frac1{T_\text{eff}}\sum_{t=K}^T
 (X_t-\alpha \bone_K) Y_t^{\obs} &= \frac1{T_\text{eff}}\sum_{t=K}^T
 \E (X_t-\alpha \bone_K) Y_t^{\obs} + o_p(1)\;, \label{eq:LLN4}\\
 \frac1{T_\text{eff}}\sum_{t=K}^T
 \E (X_t-\alpha \bone_K) Y_t^{\obs} &= \alpha (1 - \alpha) C_T \label{eq:ct_eval}\;.
\end{align}

As explained above, $X_t$ is an irreducible and aperiodic Markov chain. By standard mixing results for finite-state Markov chains and standard covariance inequalities for strongly mixing sequences, there exist constants \(C<\infty\) and \(r\in(0,1)\) such that, for any uniformly bounded functions \(f_t\) and \(f_s\),
\begin{equation}\label{eq:mixing}
\left|
\operatorname{Cov}\{f_t(X_t),f_s(X_s)\}
\right| \le C r^{|t-s|}\;.
\end{equation}
See, for example, Theorem 3.1 of \textcite{bradley2005basic} for exponential strong mixing of finite-state irreducible aperiodic Markov chains and \textcite{rio1993covariance} for the corresponding covariance inequality.

Next, under Assumption \ref{asmp:po}, $Y_t^{\obs}$ is a function of $X_t$, and therefore we may define 
\[
f_t(X_t) \coloneqq (X_t-\alpha \bone_K) Y_t^{\obs}\;,
\]
and 
\[
\frac1{T_\text{eff}}\sum_{t=K}^T
 (X_t-\alpha \bone_K) Y_t^{\obs} = \frac1{T_\text{eff}}\sum_{t=K}^T
 f_t(X_t)\;.
\]
The function satisfies $\sup_{t \ge K} \sup_{x \in \{0, 1\}^K} \|f_t(x)\| < \infty$ due to the uniform boundedness in Assumption \ref{asmp:po}, and hence the covariance bound \eqref{eq:mixing} holds for $f_{t, j}$ for a given index $j$. For any $j = 1, \dots, K$, the variance of the $j$-th component satisfies
\begin{align*}
\var\Bigl(\frac1{T_\text{eff}}\sum_{t=K}^T
 f_{t, j}(X_t)\Bigr) &= \frac{1}{T_\text{eff}^2} \sum_{t=K}^T \sum_{s = K}^T \Cov(f_{t, j}(X_t), f_{s, j}(X_s))\\
 &\le \frac{1}{T_\text{eff}^2} \sum_{t=K}^T \sum_{s = K}^T |\Cov(f_{t, j}(X_t), f_{s, j}(X_s))|\\
 &\le \frac{C}{T_\text{eff}^2} \sum_{t=K}^T \sum_{s = K}^T r^{|t-s|} \le \frac{C^\prime}{T_\text{eff}}\;,
\end{align*}
for some constant $C^\prime > 0$. Since the right hand side converges to zero, by Chebyshev's inequality, we prove Equation \eqref{eq:LLN4}. 

To prove Equation \eqref{eq:ct_eval}, we apply Assumption \ref{asmp:po} to obtain that for each entry $Z_{t-j+1}$ in $X_t$, 
\begin{align*}
\E (Z_{t-j+1}-\alpha ) Y_t^{\obs} &= \E \Bigl((1 -\alpha ) Y_t^{\obs} \mid Z_{t-j+1} = 1\Bigr) \Pr(Z_{t-j+1} = 1) + \E \Bigl((0 -\alpha ) Y_t^{\obs} \mid Z_{t-j+1} = 0\Bigr) \Pr(Z_{t-j+1} = 0)  \\
&= \alpha (1 - \alpha) \Bigl(\E (Y_t^{\obs} \mid Z_{t-j+1} = 1) - \E (Y_t^{\obs} \mid Z_{t-j+1} = 0)\Bigr) \\
&= \alpha (1 - \alpha) \Bigl(\E (Y_t(Z_{t-K+1}, \dots, Z_{t-j}, 1, Z_{t-j+2}, \dots, Z_t) \mid Z_{t-j+1} = 1) \\
& - \E (Y_t(Z_{t-K+1}, \dots, Z_{t-j}, 0, Z_{t-j+2}, \dots, Z_t) \mid Z_{t-j+1} = 0)\Bigr) \\
&= \alpha (1 - \alpha) C_{t, j}\;.
\end{align*}
For notational simplicity, the display above is written for the interior case $1<j<K$; the boundary cases $j=1$ and $j=K$ follow by the same argument. Lastly, averaging the expression above over $t = K, \dots, T$ proves Equation \eqref{eq:ct_eval}.
\end{proof}

\begin{proof}[(Proof of Corollary \ref{cor:iid})]
Since $\rho + \gamma = 1$, we have $q = 0$ and hence $\Sigma_K(q) = I_K$. Since $\rho + \gamma = 1$ implies independent treatments (Example \ref{ex:indep}), we have
\[
C_{t, k} = \E \Bigl(Y_t^{\mathrm{obs}}\mid Z_{t-k+1}=1\Bigr) - \E \Bigl(Y_t^{\mathrm{obs}}\mid Z_{t-k+1}=0\Bigr) = \E \Bigl(
Y_t\bigl(Z_{t-K+1:t}^{(t-k+1)=1}\bigr) -
Y_t\bigl(Z_{t-K+1:t}^{(t-k+1)=0}\bigr)
\Bigr)\;.
\]
Therefore, $\Sigma_K(q)^{-1} C_T = \beta^{\mathrm{LD}}$, and the claimed result holds due to Theorem \ref{thm:ols}. 
\end{proof}

\begin{proof}[(Proof of Corollary \ref{cor:additive})]
Under Assumption~\ref{asmp:additive-po}, one can verify from the definition that 
\[
\beta_k^{\mathrm{LD}} = \frac{1}{T_\text{eff}} \sum_{t=K}^T \delta_{t, k}\;.
\]
Therefore, it suffices to prove $C_T = \Sigma_K(q) \beta^{\mathrm{LD}}$. For $t = K, \dots, T$ and $i = 1, \dots, K$, we have
\[
C_{t,i}
=
\E (Y_t^{\mathrm{obs}}\mid Z_{t-i+1}=1) -
\E (Y_t^{\mathrm{obs}}\mid Z_{t-i+1}=0)\;.
\]
Using the additive linear form in Assumption \ref{asmp:additive-po}, we have
\begin{align*}
C_{t,i}
&=
\sum_{j=1}^K
\delta_{t,j}
\Bigl(
\E(Z_{t-j+1}\mid Z_{t-i+1}=1)- \E(Z_{t-j+1}\mid Z_{t-i+1}=0) \Bigr)\;.
\end{align*}
For any $i, j = 1, \dots, K$, we have 
\begin{align*}
& \E(Z_{t-j+1}\mid Z_{t-i+1}=1)- \E(Z_{t-j+1}\mid Z_{t-i+1}=0) \\
&= \Pr(Z_{t-j+1} = 1\mid Z_{t-i+1}=1) - \Pr(Z_{t-j+1} = 1\mid Z_{t-i+1}=0) \\
&= \frac{\Pr(Z_{t-j+1} = 1, Z_{t-i+1}=1)}{\Pr(Z_{t-i+1}=1)} - \frac{\Pr(Z_{t-j+1} = 1, Z_{t-i+1}=0)}{\Pr(Z_{t-i+1}=0)} \\
&= \frac{\Pr(Z_{t-j+1} = 1, Z_{t-i+1}=1)\Pr(Z_{t-i+1}=0) - \Pr(Z_{t-j+1} = 1, Z_{t-i+1}=0) \Pr(Z_{t-i+1}=1)}{\Pr(Z_{t-i+1}=1)\Pr(Z_{t-i+1}=0)} \\
&= \frac{\Pr(Z_{t-j+1} = 1, Z_{t-i+1}=1) - \Pr(Z_{t-j+1} = 1) \Pr(Z_{t-i+1}=1)}{\Pr(Z_{t-i+1}=1)\Pr(Z_{t-i+1}=0)} = q^{|i - j|}\;,
\end{align*}
where the last equality follows from Lemma \ref{lem:homog_moments}. Note that $q^{|i - j|}$ is the $(i, j)$-th entry of the matrix $\Sigma_K(q)$. We have 
\[
C_T = \Sigma_K(q) \beta^{\mathrm{LD}}\;.
\]
This completes the proof.
\end{proof}

\end{document}